\pdfoutput=1 

\documentclass[12pt]{iopart}

\expandafter\let\csname equation*\endcsname\relax
\expandafter\let\csname endequation*\endcsname\relax
\usepackage{amsmath}
\usepackage{amsthm}
\usepackage{iopams} 
\usepackage{bm}
\usepackage[svgnames,dvipsnames]{xcolor}
\usepackage[colorlinks,citecolor=DarkGreen,linkcolor=magenta,linktocpage]{hyperref}
\usepackage{physics}
\usepackage{cite}
\newtheorem{theorem}{Theorem}
\newtheorem{definition}{Definition}
\newtheorem{lemma}{Lemma}
\usepackage{tikz}
\usetikzlibrary{decorations.pathreplacing,
                calligraphy,
                matrix}

\newcommand{\mysinglebox}{
\begin{tikzpicture}[baseline=-0.3em]
        \matrix (m) [matrix of math nodes,
                 nodes={draw, 
                    minimum width=1em, 
                    minimum height=1em,
                    scale=0.8,
                    font=\scriptsize, 
                    anchor=center},
                 column sep=-\pgflinewidth,
                 row sep=-\pgflinewidth,
                 ampersand replacement=\&,
                 ]
        {
        ~ \\
        };
    \end{tikzpicture}
}

\newcommand{\myfullyantisymmetrizedbox}[1]{
\begin{tikzpicture}[
        BC/.style = {decorate,             
                decoration={calligraphic brace, amplitude=5pt, raise=1mm},
                very thick, 
                pen colour={black},
                font=\footnotesize
                    },
            baseline=-0.25em,
                    ]
        \matrix (m) [matrix of math nodes,
                 nodes={draw, 
                    minimum width=1em, 
                    minimum height=1em,
                    scale=0.8,
                    font=\scriptsize, 
                    anchor=center},
                 column sep=-\pgflinewidth,
                 row sep=-\pgflinewidth,
                 ampersand replacement=\&,
                 ]
        {
        ~ \\
        ~ \\ 
        |[draw=none, text height=3mm]| \vdots\\
        ~ \\
        };
        \draw[BC] (m-4-1.south west) -- node[left =0.5em] {#1} (m-1-1.north west);
    \end{tikzpicture}
}

\newcommand{\myantisymmetrizedandsimmetrizedbox}[1]{
\begin{tikzpicture}[
        BC/.style = {decorate,             
                decoration={calligraphic brace, amplitude=5pt, raise=1mm},
                very thick, 
                pen colour={black},
                font=\footnotesize
                    },
            baseline=-0.25em,
                    ]
        \matrix (m) [matrix of math nodes,
                 nodes={draw, 
                    minimum width=1em, 
                    minimum height=1em,
                    scale=0.8,
                    font=\scriptsize, 
                    anchor=center},
                 column sep=-\pgflinewidth,
                 row sep=-\pgflinewidth,
                 ampersand replacement=\&,
                 ]
        {
        ~ \& ~\\
        ~ \& \\ 
        |[draw=none, text height=3mm]| \vdots \& \\
        ~  \& \\
        };
        \draw[BC] (m-4-1.south west) -- node[left =0.5em] {#1} (m-1-1.north west);
    \end{tikzpicture}
}

\begin{document}
\title[Flat-band ferromagnetism in the SU(\texorpdfstring{$N$}{N}) Hubbard and Kondo lattice models]{Flat-band ferromagnetism in the SU(\texorpdfstring{$N$}{N}) Hubbard and Kondo lattice models}
\author{Kensuke Tamura}
\address{Department of Physics, The University of Tokyo, 7-3-1 Hongo, Bunkyo-ku, Tokyo 113-0033, Japan}
\ead{tamura-kensuke265@g.ecc.u-tokyo.ac.jp}

\author{Hosho Katsura}
\address{Department of Physics, The University of Tokyo, 7-3-1 Hongo, Bunkyo-ku, Tokyo 113-0033, Japan}
\address{Institute for Physics of Intelligence, The University of Tokyo,
7-3-1 Hongo, Bunkyo-ku, Tokyo 113-0033, Japan}
\address{Trans-scale Quantum Science Institute, The University of Tokyo,
7-3-1 Hongo, Bunkyo-ku, Tokyo 113-0033, Japan}
\ead{katsura@phys.s.u-tokyo.ac.jp}
\vspace{10pt}
\begin{indented}
\item[]\today
\end{indented}

\begin{abstract}
We develop a general theory of flat-band ferromagnetism in the SU($N$) Fermi-Hubbard model, which describes the behavior of $N$-component fermions with SU($N$) symmetric interactions.
We focus on the case where the single-particle spectrum has a flat band at the bottom and establish a necessary and sufficient condition for the SU($N$) Hubbard model to exhibit ferromagnetism when the number of particles is the same as the degeneracy. 
We show that the occurrence of ferromagnetism is equivalent to the irreducibility of the projection matrix onto the space of single-particle ground states.
We also demonstrate that this result can be exploited to establish a rigorous result for the ferromagnetic SU($N$) Kondo lattice model with a flat band.
Specifically, we prove that when the SU($N$) Hubbard model is ferromagnetic, the ferromagnetic SU($N$) Kondo lattice model with the same hopping matrix also exhibits SU($N$) ferromagnetism.
\end{abstract}

\vspace{2pc}
\submitto{\jpa}
%
\maketitle
%
%

\section{Introduction \label{sec:introduction}}
In recent years, advances in experimental techniques with ultracold atoms have allowed for the simulation of various quantum systems in optical lattices~\cite{bloch2008many,bloch2012quantum,lewenstein2012ultracold,gross2017quantum,schafer2020tools}.
With the ability to precisely control lattice potentials and interaction strengths, ultracold atomic systems are expected to be a versatile tool for investigating many-body physics in strongly correlated systems.
It is worth noting that ultracold atoms are not restricted to simulating known models describing conventional quantum systems but can also realize novel quantum systems with no counterpart in conventional materials.
\par 
One example of such a novel quantum system is fermionic systems with SU($N$) symmetry realized with alkaline-earth-like atoms.
These atoms trapped in an optical lattice are described by the SU($N$) Fermi-Hubbard model~\cite{gorshkov2010two}, which generalizes the standard Hubbard model with SU(2) symmetry~\cite{kanamori1963electron,gutzwiller1963effect,hubbard1963electron,montorsi1992hubbard,fazekas1999lecture}.
In conventional condensed-matter physics, the SU($N$) Hubbard model has mainly been explored with the large-$N$ approach~\cite{affleck1988large,marston1989large}.
This approach is primarily concerned with the behavior of the model with infinitely large $N$, and little attention has been paid to the properties of the model for finite $N$ ($N > 2$).
However, recent experimental realizations of the SU($N$) Hubbard model with ultracold atoms have inspired theoretical studies on the properties of the SU($N$) Hubbard model with finite $N$~\cite{taie20126,hofrichter2016direct,ozawa2018antiferromagnetic}.
Recent studies have shown that the SU($N$) Hubbard model can exhibit phases, and interest in this model has continued to grow.
\par
Moreover, in a specific limit, the system of alkaline-earth-like atoms can be described by the SU($N$) Kondo lattice model~\cite{gorshkov2010two,totsuka2023ferromagnetism}, in which itinerant fermions and localized SU($N$) spins interact with SU($N$) symmetric exchange interaction.
Efforts have also been made to realize such systems described by the SU($N$) Kondo lattice model using two-orbital alkaline-earth-like atoms~\cite{cappellini2014direct,ono2019antiferromagnetic,cappellini2019coherent}.
The SU($N$) Kondo lattice model was also introduced in the large-$N$ approach to study the SU(2) Kondo lattice model~\cite{coleman19831,read1984stability,tsunetsugu1997ground}. 
However, in this approach, the main focus was on the case with infinitely large $N$, and the properties of the models at finite $N$ have been less studied except for the case of $N=2$.
\par 
While the SU($N$) Hubbard and Kondo lattice models have attracted much interest both theoretically and experimentally, such models are notoriously difficult to solve analytically. 
Nevertheless, obtaining mathematically rigorous results in special situations would be possible. 
Although the model in such a situation may be unrealistic, it can serve as a basis for other theoretical studies. 
Here we review rigorous results for SU($N$) symmetric models, mostly for the SU($N$) Hubbard model (including the case with $N=2$).
\par
The Nagaoka ferromagnetism is the first rigorous result for the SU(2) Hubbard model~\cite{nagaoka1966ferromagnetism,tasaki1989extension}.
It was proved that when the Coulomb repulsions are infinitely large, and there is exactly one hole, the ground state of the Hubbard model is ferromagnetic and unique, provided the lattice satisfies a certain connectivity condition.
Recently, Refs.~\cite{katsura2013nagaoka,bobrow2018exact} have reported that the Nagaoka ferromagnetism can be extended to the SU($N$) Hubbard model with general $N$.
In the multiorbital Hubbard model, theorems regarding ferromagnetism have been rigorously proved in the Refs.~\cite{li2014exact,li2015exact}. 
These theorems have also been extended to the SU($N$) case, as discussed in~\cite{li2014exact}.
\par 
In Ref.~\cite{lieb1989two}, Lieb established rigorous results for both attractive and repulsive SU(2) Hubbard model. 
In particular, for the repulsive case, it was proved that Lieb's ferrimagnetism is exhibited in a wide range of models, which can also be considered as the case with a flat band. 
Subsequently, Mielke~\cite{mielke1991ferromagnetic} and Tasaki~\cite{tasaki1992ferromagnetism} independently established new rigorous results for the SU(2) Hubbard model, known as flat-band ferromagnetism. 
The term flat band refers to the structure of the single-particle energy spectrum with macroscopic degeneracy.
They constructed tight-binding models that produce a flat band at the bottom of the single-particle spectrum and then showed that the ground states of the Hubbard model are ferromagnetic and unique when the number of particles equals the multiplicity of the single-particle ground states.
There are systematic methods for constructing tight-binding models with flat bands.
For example, Mielke proposed a method based on line graphs~\cite{mielke1991ferromagnetic}.
Another method called cell construction was introduced by Tasaki~\cite{tasaki1992ferromagnetism,mielketasaki1993ferromagnetism,tasaki1998nagaoka,tasaki2020physics}.
Other methods of constructing various classes of flat bands have also been proposed~\cite{nishino2003flat,nishino2005three,hatsugai2011zq,hatsugai2015flat,katsura2015flatband,ogata2021methods}. 
Based on these methods, various types of flat-band ferromagnetism have been studied so far~\cite{sekizawa2003extension,ueda2004ferromagnetism,gulacsi2007exact,katsura2010ferromagnetism,tanaka2020extension}.
Furthermore, extensions of flat-band ferromagnetism to the SU($N$) case have recently been discussed, and rigorous results were proved in Refs.~\cite{liu2019flat,tamura2019ferromagnetism,tamura2021ferromagnetism}.
\par
Although there are various tight-binding models that have a flat band at the bottom of their energy spectrum, we should note that it is not always guaranteed that the ground state of the Hubbard model, which is formed by adding the on-site interaction term to the tight-binding model, is uniquely ferromagnetic.
In the case of the SU(2) Hubbard model, a general theory of flat-band ferromagnetism has been developed~\cite{mielke1993ferromagnetism,mielke1999stability,tasaki2020physics}.
This theory provides a necessary and sufficient condition to determine whether an SU(2) Hubbard model with a bottom flat band exhibits ferromagnetism in the ground state.
However, the corresponding general theory of flat-band ferromagnetism in the SU($N$) Hubbard model has not yet been established. 
\par 
We also comment on rigorous results for the Kondo lattice model.
As for the SU(2) Kondo lattice model, a few rigorous results are known.
In Ref.~\cite{lacroix1985some}, the equivalence between the antiferromagnetic SU(2) Kondo lattice model in the strong-coupling limit and the SU(2) Hubbard model with infinitely large Coulomb repulsion was found. 
It was shown that there exist some ferromagnetic regions.
In Ref.~\cite{sigrist1991rigorous}, the antiferromagnetic SU(2) Kondo lattice model with one electron was investigated, and it was rigorously proved that the ground state exhibits a ferromagnetic order. 
The SU(2) Kondo lattice model with a flat band was also discussed in Ref.~\cite{shen1998ferromagnetic}.
For the SU($N$) Kondo lattice model, the one-dimensional SU($N$) Kondo lattice model has recently been discussed in Ref.~\cite{totsuka2023ferromagnetism}. 
In the strong-coupling limit, the effective Hamiltonian of the model was derived, and rigorous results for the ground states were proved, which can be seen as a generalization of the result in Ref.~\cite{kubo1982note} to the SU($N$) case.
\par 
This paper presents a general theory of flat-band ferromagnetism in the SU($N$) Hubbard model.
This is a natural extension of the general theory in the SU(2) Hubbard model presented in Refs.~\cite{mielke1993ferromagnetism,mielke1999stability}.
We consider a hopping matrix whose ground states are degenerate.
Then we study the SU($N$) Hubbard model with the hopping matrix.
We give a necessary and sufficient condition for the model to exhibit SU($N$) ferromagnetism.
It is proved that the emergence of SU($N$) ferromagnetism is equivalent to the irreducibility of the orthogonal projection matrix onto the space spanned by the lowest energy states of the single-particle spectrum.
\par
In addition, we find an application of the result to the ferromagnetic SU($N$) Kondo lattice model and prove a rigorous result for flat-band ferromagnetism in this model.
The standard alkali-earth-like atoms, such as $^{87}\mathrm{Sr}$ and $^{173} \mathrm{Yb}$, exhibit the ferromagnetic Kondo coupling rather than antiferromagnetic Kondo coupling~\cite{zhang2014spectroscopic,cappellini2014direct,scazza2014observation,totsuka2023ferromagnetism}.
Consequently, in the context of ultracold atomic experiments, it is more physically natural to consider the ferromagnetic SU($N$) Kondo lattice model.
Supposing that the SU($N$) Hubbard model exhibits SU($N$) ferromagnetism, it is rigorously proved that the ferromagnetic SU($N$) Kondo lattice model with the same hopping matrix also exhibits SU($N$) ferromagnetism in its ground states.
\par 
The present paper is organized as follows.
In Sec.~\ref{sec:SU(N) hubbard model and main result}, we consider the SU($N$) Hubbard model with degenerate single-particle ground states. 
We then discuss the necessary and sufficient condition for the SU($N$) Hubbard model to exhibit ferromagnetism when the number of particles is the same as the degeneracy and prove that the irreducibility of the projection matrix onto the space of single-particle ground states is equivalent to the occurrence of ferromagnetism.
In Sec.~\ref{sec:Kondo}, we further discuss the ferromagnetic SU($N$) Kondo lattice model with a flat band. 
By exploiting the general theory for the Hubbard model, we also establish a rigorous result for flat-band ferromagnetism in the ferromagnetic SU($N$) Kondo lattice model.
Finally, in Sec.~\ref{sec:conclusion}, we give a summary and present some remarks on the theorem concerning the SU($N$) Kondo lattice model.
\section{The SU(\texorpdfstring{$N$}{N}) Hubbard model and main result \label{sec:SU(N) hubbard model and main result}}
\subsection{The SU(\texorpdfstring{$N$}{N}) Hubbard model \label{subsec:hamiltonian}}
Let $\Lambda$ be a finite lattice.
We denote by $\hat{c}_{x, \alpha}^\dag$ and $\hat{c}_{x, \alpha}$, respectively, the fermionic creation and the annihilation operators at site $x \in \Lambda$ with color $\alpha = 1, \dots, N$. 
They satisfy the anticommutation relations 
\begin{align}
    \{\hat{c}_{x, \alpha}, \hat{c}_{y, \beta}\} &= \{\hat{c}_{x, \alpha}^\dag, \hat{c}_{y, \beta}^\dag\} = 0 , \\ 
    \{\hat{c}_{x, \alpha}, \hat{c}_{y, \beta}^\dag\} &= \delta_{\alpha, \beta} \delta_{x, y}.
\end{align}
The number operator of fermion at site $x$ with color $\alpha$ is defined by $\hat{n}_{x, \alpha} = \hat{c}_{x, \alpha}^\dag \hat{c}_{x, \alpha}$, and the total fermion number is $\hat{N}_c = \sum_{x \in \Lambda} \hat{n}_x$,
where $\hat{n}_x = \sum_{\alpha=1}^N \hat{n}_{x, \alpha}$.
The Fock space of the fermionic operators is denoted by $\mathcal{H}(\Lambda)$.
The Hamiltonian of the SU($N$) Hubbard model is given by
\begin{align}
    \hat{H}_{\mathrm{Hub}} &= \hat{H}_{\mathrm{hop}} + \hat{H}_{\mathrm{int}}, \label{eq:SUn hubbard hamiltonian}\\
    \hat{H}_{\mathrm{hop}} &= \sum_{\alpha=1}^N \sum_{x, y \in \Lambda} t_{x, y} \hat{c}_{x, \alpha}^\dag \hat{c}_{y, \alpha}, \label{eq:hopping hamiltonian}\\ 
    \hat{H}_{\mathrm{int}} &= U \sum_{\alpha < \beta} \sum_{x \in \Lambda} \hat{n}_{x, \alpha} \hat{n}_{x, \beta}, \label{eq:interaction hamiltonian}
\end{align}
where $\mathsf{T} = \left(t_{x, y}\right)_{x, y \in \Lambda}$ is the hopping matrix on the lattice $\Lambda$, and the parameter $U$ is assumed to be positive.
\par
In the SU($N$) Hubbard model, the total number of fermions is trivially conserved, which can be seen as 
\begin{align}
    [\hat{H}_{\mathrm{Hub}}, \hat{N}_c] = 0.
\end{align}
We define color raising and lowering operators by
\begin{align}
    \hat{F}^{\alpha, \beta} = \sum_{x \in \Lambda} \hat{c}_{x, \alpha}^\dag \hat{c}_{x, \beta} \ \ \text{for} \ \alpha \neq \beta, 
\end{align}
and the total number operator of fermions with color $\alpha$ by
\begin{align}
    \hat{F}^{\alpha, \alpha} = \sum_{x \in \Lambda} \hat{c}_{x, \alpha}^\dag \hat{c}_{x, \alpha} \ \ \text{for} \ \alpha = 1, \dots, N.
\end{align}
Due to the SU($N$) symmetry, one can see that the operators $\hat{F}^{\alpha, \beta}$ commutes with $\hat{H}_{\mathrm{Hub}}$.
Together with the conservation of the total number of fermions, the Hamiltonian~\eqref{eq:SUn hubbard hamiltonian} possesses $\mathrm{U}(N) = \mathrm{U}(1) \times \mathrm{SU}(N)$ symmetry.
In what follows, we denote the eigenvalues of $\hat{F}^{\alpha, \alpha}$ by $M_{\alpha}$, and the eigenvalue of $\hat{N}_c$ by $N_c$.
\par
Let us introduce some subspaces of $\mathcal{H}(\Lambda)$.
We define a subspace $\mathcal{H}_{N_c}(\Lambda)$ by 
\begin{align}
    \mathcal{H}_{N_c}(\Lambda) = \{\ket{\Phi} \in \mathcal{H}(\Lambda)~|~ \hat{N}_c \ket{\Phi} = N_c \ket{\Phi}\}.
\end{align}
We also define a subspace $\mathcal{H}_{M_1, \dots, M_N}(\Lambda)$ by 
\begin{align}
    \mathcal{H}_{M_1, \dots, M_N}(\Lambda) = \{\ket{\Phi} \in \mathcal{H}(\Lambda)~|~ \hat{F}^{\alpha, \alpha} \ket{\Phi} = M_{\alpha} \ket{\Phi} \ \text{for all}\ \alpha = 1, \dots, N\}. \label{eq:Hilbert space with fixed color}
\end{align}
\par
To define SU($N$) ferromagnetism, we introduce the quadratic Casimir operator $\hat{C}_2$ of the SU($N$) group, which is defined by~\cite{ping2002group}
\begin{align}
    \hat{C}_2 = \frac{1}{2} \left(\sum_{\alpha, \beta = 1}^{N} \hat{F}^{\alpha, \beta} \hat{F}^{\beta, \alpha} - \frac{\hat{N}_c^2}{N}\right).
\end{align}
When $N=2$, in the standard notation, one may write $\hat{S}^{+}_{\mathrm{tot}} = \hat{F}^{1, 2}$, $\hat{S}^{-}_{\mathrm{tot}} = \hat{F}^{2, 1}$, and $\hat{S}_{\mathrm{tot}}^z = \left(\hat{F}^{1, 1} - \hat{F}^{2, 2}\right)/2$. 
With the notation, the operator $\hat{C}_2$ is written as 
\begin{align}
    \hat{C}_2 = \frac{1}{2} \left(\hat{S}^{+}_{\mathrm{tot}} \hat{S}^{-}_{\mathrm{tot}} + \hat{S}^{-}_{\mathrm{tot}} \hat{S}^{+}_{\mathrm{tot}}\right) + \left(\hat{S}_{\mathrm{tot}}^z\right)^2 \ \ \text{for} \ N=2.
\end{align}
This operator is the square of the magnitude of the total spin operator defined by $\left(\hat{\bm{S}}_{\mathrm{tot}}\right)^2 = \left(\hat{S}^x_{\mathrm{tot}}\right)^2 + \left(\hat{S}^y_{\mathrm{tot}}\right)^2 + \left(\hat{S}^z_{\mathrm{tot}}\right)^2$. 
Therefore, the operator $\hat{C}_2$ can be seen as a generalization of the operator $\left(\hat{\bm{S}}_{\mathrm{tot}}\right)^2$.
\par
Now we are ready to state the definition of SU($N$) ferromagnetism.
\begin{definition} 
Consider the Hamiltonian~\eqref{eq:SUn hubbard hamiltonian} with the total fermion number $N_c$.
We say that the model exhibits SU($N$) ferromagnetism if any ground state $\ket{\Phi_{\mathrm{GS}}}$ has the maximum eigenvalue of $\hat{C}_2$ in $\mathcal{H}_{N_c}(\Lambda)$, i.e.,
\begin{align}
    \hat{C}_2 \ket{\Phi_{\mathrm{GS}}} = \frac{N_c (N-1)}{2} \left(\frac{N_c}{N} + 1\right) \ket{\Phi_{\mathrm{GS}}}. \label{eq:def of ferromagnetic state}
\end{align}
\end{definition}
Note that the above definition of ferromagnetism is the strongest form of ferromagnetism, which should be referred to as complete ferromagnetism or saturated ferromagnetism. 
In the case $N=2$, let $S_{\mathrm{tot}} \left(S_{\mathrm{tot}} + 1\right)$ be the eigenvalue of $\left(\hat{\bm{S}}_{\mathrm{tot}}\right)^2$. 
Even if $S_{\mathrm{tot}}$ of the ground states is macroscopically large but not the maximum value, it is commonly considered that ferromagnetism is manifested. 
In this paper, however, we only study complete ferromagnetism and refer to it simply as ferromagnetism.
\par 
We call a state satisfying Eq.~\eqref{eq:def of ferromagnetic state} as a fully polarized state. 
The eigenvalue equation \eqref{eq:def of ferromagnetic state} is satisfied if a state has no double occupancy and is fully symmetrized with respect to the color degrees of freedom. 
Conversely, a state satisfying Eq.~\eqref{eq:def of ferromagnetic state} is such a fully symmetrized state. 
\subsection{Main theorem \label{subsec:main theorem}}
Here we state our main theorem.
First, we introduce some notation and make assumptions about the model.
The single-particle Hilbert space is denoted by $\mathfrak{h} \cong \mathbb{C}^{|\Lambda|}$, and we write a $|\Lambda|$-dimensional vector in $\mathfrak{h}$ as $\bm{\phi} = (\phi(x))_{x \in \Lambda}$.
The inner product of two vectors, $\bm{\phi}$ and $\bm{\psi}$, is defined by 
\begin{align}
    \langle \bm{\phi}, \bm{\psi} \rangle = \sum_{x \in \Lambda} \phi(x)^{*} \psi(x).
\end{align}
We now assume that 
\begin{align}
    \mathsf{T} \ge 0,
\end{align}
and denote the kernel of $\mathsf{T}$ by $\mathfrak{h}_0 = \mathrm{ker} \mathsf{T}$.
We also assume that $\mathfrak{h}_0$ is not empty and write $D_0 = \mathrm{dim} \mathfrak{h}_0$.
Let $\mathsf{P}_0$ be the orthogonal projection matrix onto the subspace $\mathfrak{h}_0$, and we define $\Lambda_0 = \{x \in \Lambda| (\mathsf{P}_0)_{x, x} \neq 0\}$.
We say the $|\Lambda_0| \times |\Lambda_0|$ matrix $\left(\left(\mathsf{P}_0\right)_{x, y}\right)_{x, y \in \Lambda_0}$ is reducible if and only if $\Lambda_0$ can be decomposed as $\Lambda_0 = \Lambda_1 \cup \Lambda_2$ with $\Lambda_1 \cap \Lambda_2 = \emptyset$, $\Lambda_0 \neq \emptyset$, and $\Lambda_2 \neq \emptyset$ so that $\left(\mathsf{P}_0\right)_{x, y} = 0$ for any $x \in \Lambda_1$ and $y \in \Lambda_2$.
The matrix $\left(\left(\mathsf{P}_0\right)_{x, y}\right)_{x, y \in \Lambda_0}$ is said to be irreducible if it is not reducible.
If $\mathsf{T}$ has translation symmetry, we have energy bands as a function of wave vectors.
Moreover, if $D_0$ is proportional to the number of sites $|\Lambda|$, it suggests that the lowest band is flat at zero energy.
Figure~\ref{fig:delta chain} shows an example of a lattice system in which the lowest band is flat, called a delta chain. 
The realization of this lattice system with an optical lattice is also discussed~\cite{zhang2015one}.
The dispersion relations of this lattice system are shown in Fig.~\ref{fig:energy band}.
\begin{figure}[t]
    \centering
    \includegraphics[width = 0.4\columnwidth]{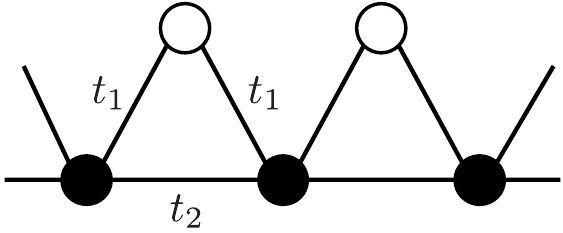}
    \caption{
    The lattice structure of the delta chain with hopping amplitude $t_1 = t/\sqrt{2}$ and $t_2 = t/2$, where we impose the periodic boundary conditions.
    All the sites have the uniform on-site potentials $t$.
    } \label{fig:delta chain}
\end{figure}

\begin{figure}
    \centering
    \includegraphics[width = 0.5\columnwidth]{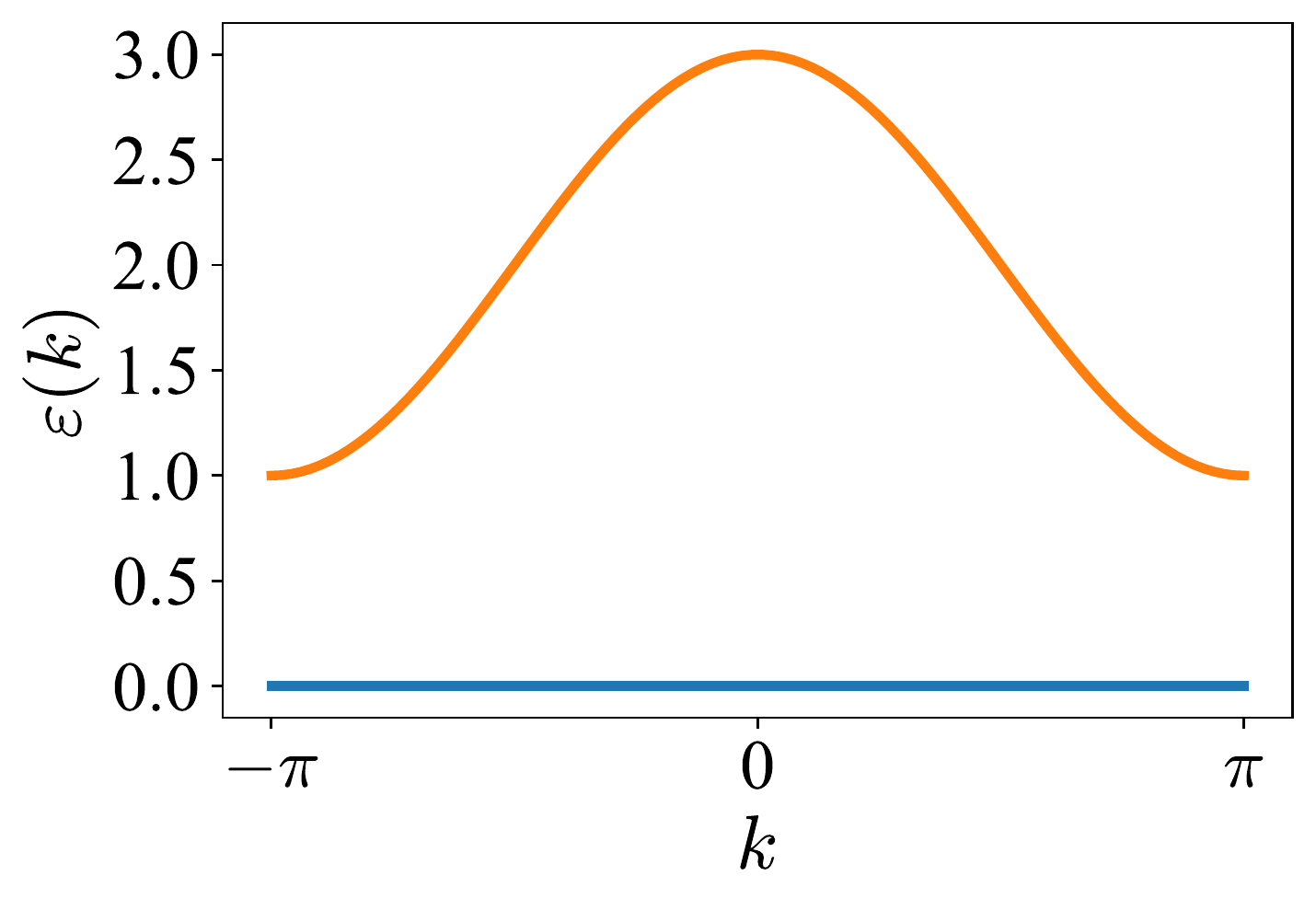}
    \caption{
    The dispersion relations of the energy bands for the delta chain with $t=1$.
    The lowest band is completely flat at zero energy
    In this case, the value of $D_0$ equals the number of the unit cells.
    }\label{fig:energy band}
\end{figure}

Now we are ready to state our theorem.
\begin{theorem} \label{thm:SUn flat band ferro}
Consider the SU($N$) Hubbard model~\eqref{eq:SUn hubbard hamiltonian} with $\mathsf{T} \geq 0$ and $N_c = D_0$ .
The model exhibits SU($N$) ferromagnetism if and only if the $|\Lambda_0| \times |\Lambda_0|$ matrix $\left(\left(\mathsf{P}_0\right)_{x, y}\right)_{x, y \in \Lambda_0}$ is irreducible.
\end{theorem}
\subsection{Proof of Theorem~\ref{thm:SUn flat band ferro} \label{subsec:proof of thm 1}}
In this subsection, we prove Theorem~\ref{thm:SUn flat band ferro}.
We first prove the following lemma.
\begin{lemma}\label{lemma:z basis}
    One can take a subset $I \subset \Lambda$ with $|I| = D_0$ and a basis $\{\bm{\mu}_z\}_{z \in I}$ of $\mathfrak{h}_0$ in such a way that for each $z \in I$, the basis vector $\bm{\mu}_z = \left(\mu_z(x)\right)_{x \in \Lambda}$ satisfies $\mu_z(z) \neq 0$ and $\mu_z(z') = 0$ for any $z' \in I\backslash \{z\}$.
\end{lemma}
\begin{proof}
Our proof is essentially the same as the proof of Lemma 11.16 in Ref.~\cite{tasaki2020physics}.
We see that the rank of $\mathsf{T}$ is $|\Lambda| - D_0$ since $\dim{\mathfrak{h}_0} = D_0$.
Then there exists a subset $\Lambda' \subset \Lambda$ with $|\Lambda'| = |\Lambda| - D_0$ such that the determinant of the submatrix $\left(t_{x, y}\right)_{x, y \in \Lambda'}$ is nonzero, and any $(|\Lambda| - D_0 + 1) \times (|\Lambda| - D_0 + 1)$ submatrix of $\mathsf{T}$ has determinant zero~\cite{prasolov1994problems}.
Let $I = \Lambda \backslash \Lambda'$. 
We find that, for arbitrary $z \in I$, the submatrix $\left(t_{x, y}\right)_{x, y \in \Lambda' \cup \{z\}}$ has determinant zero, which implies that this matrix has a zero eigenvalue.
We denote the corresponding eigenvector by $\bm{\tilde{\mu}}_z = \left(\tilde{\mu}_z(x)\right)_{x \in \Lambda' \cup \{z\}}$.
We can see that $\tilde{\mu}_z(z) \neq 0$.
This is because if $\tilde{\mu}_z(z) = 0$, then $\left(\tilde{\mu}_z(x)\right)_{x \in \Lambda'}$ is an eigenvector of $\left(t_{x, y}\right)_{x, y \in \Lambda'}$ with eigenvalue zero.
This contradicts that the matrix $\left(t_{x, y}\right)_{x, y \in \Lambda'}$ has nonzero determinant.
Thus, we have $\mu_z(z) \neq 0$.
We then define a $|\Lambda|$-dimensional vector $\bm{\mu}_z = \left(\mu_z(x)\right)_{x \in \Lambda}$ as 
\begin{align}
    \mu_z(x) = \begin{cases}
        \tilde{\mu}_z(x) \ \ \text{if} \ x \in \Lambda' \cup \{z\}, \\
        0 \ \ \text{otherwise},
    \end{cases}
\end{align}
for $z \in I$.
We note that $\mu_z(z) \neq 0$ for $z \in I$ and $\mu_z(z') = 0$ for $z' \in I \backslash \{z\}$.
Using $\sum_{y \in \Lambda' \cup \{z\}} t_{x, y} \tilde{\mu}_z(y) = 0$, we can see that 
\begin{align}
    \langle \bm{\mu}_z, \mathsf{T} \bm{\mu}_z \rangle = 0.
\end{align}
Because of the positive semidefiniteness of $\mathsf{T}$, it holds that $\mathsf{T} \bm{\mu}_z = 0$.
Since $\mu_z(z) \neq 0$ and $\mu_z(z') = 0$, the set $\{\bm{\mu}_z\}_{z \in I}$ is linearly independent, and hence it is a basis of $\mathfrak{h}_0$ satisfying the conditions of Lemma~\ref{lemma:z basis}.
\end{proof}
In the example shown in Fig.~\ref{fig:delta chain}, the subset $I$ can be taken as the entire set of the black sites.
In this case, the vector $\bm{\mu}_z$ is localized at a black site.
This vector has nonzero components only at the black site and the two white sites adjacent to it.
See Fig.~\ref{fig:zero-energy state}

\begin{figure}
    \centering
    \includegraphics[width=0.4\columnwidth]{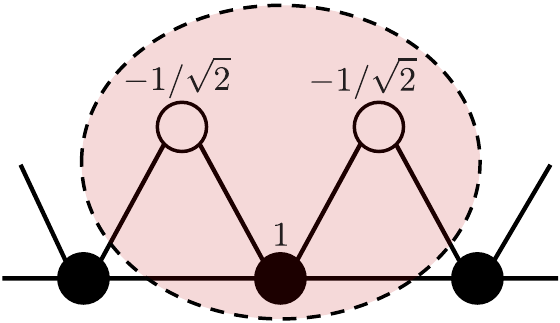}
    \caption{In the delta chain, the vector $\bm{\mu}_z$ satisfying the condition of Lemma~\ref{lemma:z basis} can be defined to be localized at the black site.
    This vector has a component of $1$ at the black site $z$ and $-1/\sqrt{2}$ at the two white sites adjacent to $z$.}\label{fig:zero-energy state}
\end{figure}

With the basis $\{\bm{\mu}_z\}_{z \in I}$, we can characterize $\Lambda_0$ as 
\begin{align}
    \Lambda_0 = \{x \in \Lambda| \mu_z(x) \neq 0 \ \ \text{for some} \ z \in I\}. \label{eq:characterization of Lambda_0}
\end{align}
This can be seen as follows.
Let $\{\bm{\psi}_i\}_{i=1, \dots, D_0}$ be an orthonormal basis of $\mathfrak{h}_0$.
The projection matrix $\mathsf{P}_0$ is written as $\left(\mathsf{P}_0\right)_{x, y} = \sum_{i=1}^{D_0} \psi_i(x) \psi_i(y)^*$.
Suppose that $\mu_z(x) = 0$ for all $z \in I$ for some $x \in \Lambda$.
Then we see that $\psi_i(x) = 0$ for all $i = 1, \dots, D_0$ since the vector $\bm{\psi}_i$ can be written as a linear combination of $\{\bm{\mu}_z\}_{z \in I}$.
Therefore, we have $\left(\mathsf{P}_0\right)_{x, x} = \sum_{i=1}^{D_0} \psi_i(x) \psi_i(x)^* = 0$, which means that $x \notin \Lambda_0$ if $\mu_z(x) = 0$ for all $z \in I$.
Conversely, suppose that for some $x \in \Lambda$, there exists $z \in I$ such that $\mu_z(x) \neq 0$.
Let $\mathsf{P}_z$ be the projection matrix onto the one-dimensional subspace spanned by $\bm{\mu}_z$.
Then we see that $\left(\mathsf{P}_0\right)_{x, x} \geq \left(\mathsf{P}_z\right)_{x, x}$.
Since $\mu_z(x) \neq 0$, $\left(\mathsf{P}_z\right)_{x, x} > 0$.
Therefore, we have $\left(\mathsf{P}_0\right)_{x, x} \neq 0$, and hence, $x \in \Lambda_0$.
Thus, $\Lambda_0 = \{x \in \Lambda| \mu_z(x) \neq 0 \ \ \text{for some} \ \ z \in I\}$.
For example, in the delta chain, the set of the eigenvectors with eigenvalue zero $\{\bm{\mu}_z\}_{z \in I}$ can cover the entire lattice system, thus $\Lambda_0 = \Lambda$.
\par
We write $\bm{\mu}_z \sim \bm{\mu}_{z'}$ if there is a site $x \in \Lambda$ such that $\mu_z(x) \mu_{z'}(x) \neq 0$.
We say that the basis $\{\bm{\mu}_{z}\}_{z \in I}$ is connected if there is a sequence $\{z_{i}\}_{i=0, \dots, n}$ with $z_i \in I$ such that $z_0 = z$, $z_n = z'$, and $\bm{\mu}_{z_{i-1}} \sim \bm{\mu}_{z_i}$ for $i=1, \dots, n$.
We also write $\bm{\mu}_z \nsim \bm{\mu}_{z'}$ if there is no site $x$ such that $\mu_z(x) \mu_{z'}(x) \neq 0$. 
\par 
Then we can prove the following lemma.
\begin{lemma}\label{lemma:SUn ferromagnetism and connectivity}
    Consider the SU($N$) Hubbard model with $N_c = D_0$.
    The model exhibits the SU($N$) ferromagnetism if and only if the basis $\{\bm{\mu}_z\}_{z \in I}$ is connected.
\end{lemma}
\begin{proof}
Since the hopping matrix $\mathsf{T}$ is positive semidefinite, $\hat{H}_{\mathrm{hop}}$ is also positive semidefinite.
This can be seen as follows.
Since $\dim{\mathfrak{h}} = |\Lambda|$ and $\dim{\mathfrak{h}_0} = D_0$, there are $|\Lambda| - D_0$ linearly independent eigenvectors of $\mathsf{T}$ with positive eigenvalues.
We denote the eigenvectors by $\bm{\phi}_i$ ($i = 1, \dots, |\Lambda| - D_0$), which satisfy $\mathsf{T} \bm{\phi}_i = \lambda_i \bm{\phi}_i$ with $\lambda_i > 0$.
Since they can always be taken to be orthogonal to each other, we assume that $\langle \bm{\phi}_i, \bm{\phi}_j \rangle = \delta_{i, j}$.
With the vectors $\bm{\mu}_z$ and $\bm{\phi}_i$, we define a new set of operators
\begin{align}
    \hat{a}_{z, \alpha}^\dag &= \sum_{x \in \Lambda} \mu_z(x) \hat{c}_{x, \alpha}^\dag, \label{eq:a operator} \\
    \hat{b}_{i, \alpha}^\dag &= \sum_{x \in \Lambda} \phi_i(x) \hat{c}_{x, \alpha}^\dag. \label{eq:b operator}
\end{align}
They satisfy 
\begin{align}
    \{\hat{a}_{z, \alpha}, \hat{a}_{w, \beta}\} &= \{\hat{b}_{i, \alpha}, \hat{b}_{j, \beta}\} = \{\hat{a}_{z, \alpha}, \hat{b}_{i, \beta}\} = 0, \\
    \{a_{z, \alpha}, a_{w, \beta}^\dag\} &= \delta_{\alpha, \beta} \langle \bm{\mu}_{z}, \bm{\mu}_{w}\rangle, \\ 
    \{\hat{b}_{i, \alpha}, \hat{b}_{j, \beta}^\dag\} &= \delta_{\alpha, \beta} \delta_{i, j}, \\ 
    \{\hat{a}_{z, \alpha}, \hat{b}_{i, \beta}^\dag\} &= 0, \label{eq:ac relation a and b}
\end{align}
where the last line follows since $\langle \bm{\mu}_z, \bm{\phi}_i\rangle = 0$.
Because the hopping matrix can be written as $t_{x, y} = \sum_{i=1}^{|\Lambda| - D_0} \lambda_i \phi_i(x) \phi_i(y)^*$, we can represent the hopping Hamiltonian as 
\begin{align}
    \hat{H}_{\mathrm{hop}} = \sum_{\alpha = 1}^N \sum_{i=1}^{|\Lambda| - D_0} \lambda_i \hat{b}_{i, \alpha}^\dag \hat{b}_{i, \alpha}. \label{eq:hopping Hamiltonian bdagb}
\end{align}
Since $\hat{b}_{i, \alpha}^\dag \hat{b}_{i, \alpha} \geq 0$ and $\lambda_i > 0$, $\hat{H}_{\mathrm{hop}}$ is positive semidefinite.
The interaction term $\hat{H}_{\mathrm{int}}$ is also positive semidefinite because $\hat{n}_{x, \alpha} \hat{n}_{x, \beta} = (\hat{c}_{x, \alpha} \hat{c}_{x, \beta})^\dag \hat{c}_{x, \alpha} \hat{c}_{x, \beta} \geq 0$.
Hence, the entire Hamiltonian $\hat{H}_{\mathrm{Hub}}$ is positive semidefinite.
\par
We define the fully polarized states as follows. 
First, the fully polarized state with color $\alpha$ is given by 
\begin{align}
    \ket{\Phi_{\mathrm{all} \ \alpha}} = \left(\prod_{z \in I} \hat{a}_{z, \alpha}^\dag\right) \ket{\Phi_{\mathrm{vac}}}, \label{eq:all alpha}
\end{align}
where $\ket{\Phi_{\mathrm{vac}}}$ is the normalized vacuum state for $\hat{c}_{x, \alpha}$.
We can easily see that the $\ket{\Phi_{\mathrm{all} \ \alpha}}$ is an eigenstate of $\hat{H}_{\mathrm{Hub}}$ with zero energy.
Since $\hat{H}_{\mathrm{Hub}} \ge 0$, the state $\ket{\Phi_{\mathrm{all} \ \alpha}}$ is the ground state of $\hat{H}_{\mathrm{Hub}}$ in $\mathcal{H}_{D_0}(\Lambda)$.
Due to the SU($N$) symmetry, one can obtain other ground states of the following form
\begin{align}
    \ket{\Phi_{M_1, M_2, \dots, M_N}} = \left(\hat{F}^{N, 1}\right)^{M_N} \cdots \left(\hat{F}^{2, 1}\right)^{M_2} \ket{\Phi_{\mathrm{all}\ 1}}, \label{eq:fully polarized state}
\end{align}
where $M_1 = D_0 - \sum_{\alpha = 2}^{N} M_{\alpha}$.
We also refer to the states of the form~\eqref{eq:fully polarized state} as fully polarized states.
It is easily seen that 
\begin{align}
    \hat{C}_2 \ket{\Phi_{\mathrm{all}\ \alpha}} = \frac{D_0 (N-1)}{2} \left(\frac{D_0}{N} + 1\right) \ket{\Phi_{\mathrm{all} \ \alpha}}, \label{eq:phi all alpha has maximum C2}
\end{align}
which means that the state $\ket{\Phi_{\mathrm{all} \ \alpha}}$ has the maximum eigenvalue of $\hat{C}_2$ in $\mathcal{H}_{D_0}(\Lambda)$.
Because $[\hat{C}_2, \hat{F}^{\alpha, \beta}] = 0$, the fully polarized states of the form~\eqref{eq:fully polarized state} also have the same eigenvalue for $\hat{C}_2$.
Thus, all the fully polarized states are ground states of $\hat{H}_{\mathrm{Hub}}$ with the maximum eigenvalue of $\hat{C}_2$ in $\mathcal{H}_{D_0}(\Lambda)$.
\par
In the following, we prove that there are no other ground states if and only if $\{\bm{\mu}_z\}_{z \in I}$ is connected.
Let $\ket{\Phi_{\mathrm{GS}}}$ be an arbitrary ground state of $\hat{H}_{\mathrm{Hub}}$ in $\mathcal{H}_{D_0}(\Lambda)$.
In general, we can express the state as
\begin{align}
    \ket{\Phi_{\mathrm{GS}}} & = \sum_{\substack{I_1, \dots, I_N \subset I \\ \tilde{I}_1, \dots, \tilde{I}_N \subset \tilde{I}}} 
    f\left(\{I_{\alpha} \}, \{\tilde{I}_{\alpha}\} \right) \times \nonumber \\
    & \left(\prod_{z_1 \in I_1} \hat{a}_{z_1, 1}^\dag \right) \cdots
    \left(\prod_{z_N \in I_N} \hat{a}_{z_N, N}^\dag\right)
    \left(\prod_{i_1 \in \tilde{I}_1} \hat{b}_{i_1, 1}^\dag\right) \cdots 
    \left(\prod_{i_N \in \tilde{I}_N} \hat{b}_{i_N, N}^\dag\right) \ket{\Phi_{\mathrm{vac}}},
\end{align}
where $f(\{I_{\alpha}\}, \{\tilde{I}_{\alpha}\})$ is a coefficient, and $I_{\alpha}$ and $\tilde{I}_{\alpha}$ are subsets of $I$ and $\tilde{I} = \{1, 2, \dots, |\Lambda| - D_0\}$, respectively, such that $\sum_{\alpha=1}^N \left(|I_{\alpha}| + |\tilde{I}_{\alpha}|\right) = D_0$.
The ground state satisfies $\hat{H}_{\mathrm{Hub}} \ket{\Phi_{\mathrm{GS}} }= 0$, and the inequalities $\hat{H}_{\mathrm{hop}} \geq 0$ and $\hat{H}_{\mathrm{int}} \geq 0$ imply that $\hat{H}_{\mathrm{hop}} \ket{\Phi_{\mathrm{GS}}} = 0$ and $\hat{H}_{\mathrm{int}} \ket{\Phi_{\mathrm{GS}}} = 0$.
Since the hopping Hamiltonian can be expressed as in Eq.~\eqref{eq:hopping Hamiltonian bdagb} and $\hat{b}_{i, \alpha}^\dag \hat{b}_{i, \alpha} \geq 0$ for all $ i = 1, \dots, |\Lambda| - D_0$ and $\alpha$, the condition $\hat{H}_{\mathrm{hop}} \ket{\Phi_{\mathrm{GS}}} = 0$ leads to 
\begin{align}
    \hat{b}_{i, \alpha} \ket{\Phi_{\mathrm{GS}}} = 0 \ \ \text{for all} \ i=1, \dots, |\Lambda|- D_0 \ \text{and} \ \alpha = 1, \dots, N.
    \label{eq:b=0}
\end{align}
Similarly, the equation $\hat{H}_{\mathrm{int}} \ket{\Phi_{\mathrm{GS}}} = 0$ reduces to 
\begin{align}
    \hat{c}_{x, \alpha} \hat{c}_{x, \beta} \ket{\Phi_{\mathrm{GS}}} = 0 \ \ \text{for any} \ x \in \Lambda  \ \text{and} \ \alpha \neq \beta.
    \label{eq:int=0}
\end{align}
Here we consider the condition~\eqref{eq:b=0}.
Noting the anticommutation relations~\eqref{eq:ac relation a and b}, we find that the ground state $\ket{\Phi_{\mathrm{GS}}}$ consists only of $\hat{a}_{z, \alpha}^\dag$ operators, i.e., the state $\ket{\Phi_{\mathrm{GS}}}$ is written as 
\begin{align}
    \ket{\Phi_{\mathrm{GS}}} = \sum_{I_1, \dots, I_N \subset I} g(\{I_\alpha \}) \left(\prod_{z_1 \in I_1} \hat{a}_{z_1, 1}^\dag \right) 
    \cdots
    \left(\prod_{z_N \in I_N} \hat{a}_{z_N, N}^\dag \right) \ket{\Phi_{\mathrm{vac}}}, 
\end{align}
where $g(\{I_{\alpha}\})$ is a coefficient, and $\sum_{\alpha = 1}^{N} |I_{\alpha}| = D_0$.
\par 
Then we examine Eq.~\eqref{eq:int=0}.
We first consider the case where $x = z \in I$.
In general, it holds that 
\begin{align}
    \{\hat{c}_{x, \alpha}, \hat{a}_{z, \beta}^\dag\} = \delta_{\alpha, \beta} \mu_{z}(x), \label{eq:anticommutation rel c and a}
\end{align}
for any $x \in \Lambda$ and $z \in I$.
In particular, when $x = z \in I$, we have the following anticommutation relations 
\begin{align}
    \{\hat{c}_{z, \alpha}, \hat{a}_{z', \beta}^\dag\} = \delta_{\alpha, \beta} \delta_{z, z'} \mu_z(z) \label{eq:anticommutation rel c and a, x=z}
\end{align}
for all $z, z' \in I $.
Using the anticommutation relations~\eqref{eq:anticommutation rel c and a, x=z}, we find that $g(\{I_{\alpha}\}) = 0$ if there is a pair of colors $\alpha$ and $\beta$ such that $I_{\alpha} \cap I_{\beta} \neq \emptyset$.
Since $\sum_{\alpha = 1}^N |I_{\alpha}| = D_0$, we have $\cup_{\alpha =1}^N I_{\alpha} = I$ when $I_{\alpha} \cap I_{\beta} = \emptyset$.
Therefore, the ground state takes the form
\begin{align}
    \ket{\Phi_{\mathrm{GS}}} = \sum_{\bm{\alpha}} C(\bm{\alpha}) \left(\prod_{z \in I} \hat{a}_{z, \alpha_z}^\dag\right) \ket{\Phi_{\mathrm{vac}}},
\end{align}
where $\bm{\alpha} = (\alpha_{z})_{z \in I}$ represents a color configuration over $I$, and the sum is taken over all possible color configurations.
\par 
We next consider Eq.~\eqref{eq:int=0} for $x \in \Lambda \backslash I$.
With the use of Eq.~\eqref{eq:anticommutation rel c and a}, the condition~\eqref{eq:int=0} yields 
\begin{align}
    &\sum_{z_1 < z_2} \sum_{
    \substack{\bm{\gamma} \\ \gamma_{z_1} = \beta, \gamma_{z_2} = \alpha}} \mathrm{sgn}(z_1, z_2; I) \mu_{z_1}(x) \mu_{z_2}(x) \nonumber \\
    &\times \left(C(\bm{\gamma}) - C(\bm{\gamma}_{z_1 \leftrightarrow z_2})\right) \left(\prod_{z \in I \backslash \{z_1, z_2\}} \hat{a}_{z, \gamma_z}^\dag\right) \ket{\Phi_{\mathrm{vac}}} = 0,
\end{align}
where $\mathrm{sgn}(z_1, z_2; I)$ is a sign factor arising from exchanges of the fermion operators.
The configuration $\bm{\gamma}_{z_1 \leftrightarrow z_2}$ is obtained from $\bm{\gamma}$ by swapping $\gamma_{z_1}$ and $\gamma_{z_2}$.
Since all the states in the sum are linearly independent, we find that 
\begin{align}
    \mu_{z_1}(x) \mu_{z_2}(x) \left(C(\bm{\gamma}) - C(\bm{\gamma}_{z_1 \leftrightarrow z_2})\right) = 0.
\end{align}
for all $z_1, z_2 \in I$ such that $z_1 \neq z_2$ and $\bm{\gamma}$.
When $z_1$ and $z_2$ satisfy $\bm{\mu}_{z_1} \sim \bm{\mu}_{z_2}$, then by definition, there is a site $x \in \Lambda \backslash I$ such that $\mu_{z_1}(x) \mu_{z_2}(x) \neq 0$.
Consequently, we have $C(\bm{\gamma}) = C(\bm{\gamma}_{z_1 \leftrightarrow z_2})$ if $\bm{\mu}_{z_1} \sim \bm{\mu}_{z_2}$.
When the basis $\{\bm{\mu}_z\}_{z \in I}$ is connected, for any $z, z' \in I$, we can take a sequence $z_1=z, z_2, \dots, z_n=z'$ such that $\bm{\mu}_{z_i} \sim \bm{\mu}_{z_{i+1}}$.
Thus, we see $C(\bm{\gamma}) = C(\bm{\gamma}_{z \leftrightarrow z'})$ for any $z, z' \in I$.
By noting that any permutation of $\bm{\alpha}$ can be obtained by repeatedly swapping two colors, we have 
\begin{align}
    C(\bm{\alpha}) = C(\bm{\beta}), \label{eq:Calpha = Cbeta}
\end{align}
where $\bm{\beta}$ is any color configuration obtained as a permutation of $\bm{\alpha}$.
As a result, when $M_{\alpha}$ in Eq.~\eqref{eq:Hilbert space with fixed color} is fixed so that $\sum_{\alpha=1}^N M_{\alpha} = D_0$, the ground state is unique in $\mathcal{H}_{M_1, \dots, M_N}(\Lambda)$.
\par
We can verify that the states satisfying Eq.~\eqref{eq:Calpha = Cbeta} are actually the fully polarized states defined by Eqs.~\eqref{eq:all alpha} and~\eqref{eq:fully polarized state}. 
To see this, we use a concept of a word.
A word $w = (w_1, \dots, w_{D_0})$ is a sequence in which $w_i \in \{1, \dots, N\}$ for all $i=1, \dots, D_0$.
We denote the number of occurrences of $\alpha$ in $w$ by $|w_{\alpha}|$.
The set of words for which $|w_{\alpha}| = M_{\alpha}$ is defined by 
$W(M_1, \dotsm M_N) = \{w \,| \, |w_{\alpha}| = M_{\alpha}\ \text{for all}\ \alpha\}$.
For example, $W(2, 0, 1)$ consists of $(1, 1, 3)$, $(1, 3, 1)$, and $(3, 1, 1)$.
With the notation, the ground state in $\mathcal{H}_{M_1, \dots, M_N}(\Lambda)$ satisfying Eq.~\eqref{eq:Calpha = Cbeta} is written as 
\begin{align}
    \ket{\tilde{\Phi}_{M_1, \dots, M_N}} = \sum_{w \in W(M_1, \dots, M_N)} \hat{a}_{z_1, w_1}^\dag \hat{a}_{z_2, w_2}^\dag \cdots \hat{a}_{z_{D_0}, w_{D_0}}^\dag \ket{\Phi_{\mathrm{vac}}},
\end{align}
where we have denoted the set $I$ by $I = \{z_1, z_2, \dots, z_{D_0}\}$.
Noting the commutation relations $[\hat{F}^{\alpha, \beta}, \hat{a}_{z, \gamma}^\dag] = \delta_{\alpha, \gamma} \hat{a}_{z, \beta}^\dag$, we see that 
\begin{align}
    (\hat{F}^{2, 1})^{M_2} \ket{\Phi_{\mathrm{all} \ 1}} = M_2 ! \sum_{w \in W(D_0 - M_2, M_2, 0, \dots, 0)} \hat{a}_{z_1, w_1}^\dag \hat{a}_{z_2, w_2}^\dag \cdots \hat{a}_{z_{D_0}, w_{D_0}}^\dag \ket{\Phi_{\mathrm{vac}}}.
\end{align}
By repeating the same procedure, we obtain $\ket{\tilde{\Phi}_{M_1, \dots, M_N}}$ is the same as $\ket{\Phi_{M_1, \dots, M_N}}$ up to a normalization.
Therefore, the unique ground state $\ket{\tilde{\Phi}_{M_1, \dots, M_N}}$ is the fully polarized state.
\par 
If the basis $\{\bm{\mu}\}_{z \in I}$ is not connected, we can construct a different ground state other than fully polarized states.
Hence, the SU($N$) Hubbard model exhibits the SU($N$) ferromagnetism if and only if the basis $\{\bm{\mu}_{z}\}_{z \in I}$ is connected.
\end{proof}
The degeneracy of the ground states does not depend on how we take $I$ and $\{\bm{\mu}_z\}_{z \in I}$.
Hence, the connectivity of the basis characterized by Lemma~\ref{lemma:z basis} is independent of the choice of $I$ and $\{\bm{\mu}_z\}_{z \in I}$.
\par
Finally, we show that the irreducibility of $\left(\left(\mathsf{P}_0\right)_{x, y}\right)_{x, y \in \Lambda_0}$ is equivalent to the connectivity of $\{\bm{\mu}_{z}\}_{z \in I}$.
\begin{lemma} \label{lemma:connectivity and irreducibility}
    The matrix $\left(\left(\mathsf{P}_0\right)_{x, y}\right)_{x, y \in \Lambda_0}$ is irreducible if and only if $\{\bm{\mu}_{z}\}_{z \in I}$ introduced in Lemma~\ref{lemma:z basis} is connected.
\end{lemma}
\begin{proof}
We first show that non-connectivity of $\{\bm{\mu}_z\}_{z \in I}$ leads to reducibility of $\left(\left(\mathsf{P}_0\right)_{x, y}\right)_{x, y \in \Lambda_0}$. 
Let $I$ and $\{\bm{\mu}_{z}\}_{z \in I}$ be a subset and a basis characterized by Lemma~\ref{lemma:z basis}.
We assume that $I$ can be decomposed as $I = I_1 \cup I_2$ with $I_1 \cap I_2 = \emptyset$, $I_1 \neq \emptyset$ and $I_2 \neq \emptyset$, in which $\bm{\mu}_z \nsim \bm{\mu}_{z'}$ for any $z \in I_1$ and $z' \in I_2$.
When $\bm{\mu}_z \nsim \bm{\mu}_{z'}$,  it holds that $\langle\bm{\mu}_z, \bm{\mu}_{z'} \rangle = 0$ since $\mu_z(x) \mu_{z'}(x) = 0$ for any $x \in \Lambda$.
We then define subsets $\Lambda_1$ and $\Lambda_2$ as 
\begin{align}
    \Lambda_j = \{x \in \Lambda| \mu_z(x) \neq 0 \ \text{for some}\ z \in I_j\}\ \ \ \text{for} \ \ j=1, 2.
\end{align}
Because $\bm{\mu}_z \nsim \bm{\mu}_{z'}$ for any $z \in I_1$ and $z' \in I_2$, we see that $\Lambda_1 \cap \Lambda_2 = \emptyset$, and obviously $\Lambda_0 = \Lambda_1 \cup \Lambda_2$.
By taking a linear combination of $\{\bm{\mu}_z\}_{z \in I_1}$, we obtain an orthonormal set $\{\bm{\psi}^{(1)}_{i}\}_{i=1, \dots, |I_1|}$ which spans the same space spanned by $\{\bm{\mu}_{z}\}_{z \in I_1}$.
Similarly, we can take an orthonormal set $\{\bm{\psi}^{(2)}_{j}\}_{j=1, \dots, |I_2|}$ which spans the same space spanned by $\{\bm{\mu}_{z'}\}_{z' \in I_2}$.
The subspace $\mathfrak{h}_0$ is spanned by the orthonormal basis $\{\bm{\psi}^{(1)}_{i}\}_{i=1, \dots, |I_1|} \cup \{\bm{\psi}^{(2)}_{j}\}_{j=1, \dots, |I_2|}$, and the projection matrix $\mathsf{P}_0$ can be written as 
\begin{align}
    \left(\mathsf{P}_0\right)_{x, y} = \sum_{i=1}^{|I_1|} \psi^{(1)}_i(x) \left(\psi^{(1)}_i(y)\right)^{*} 
    + \sum_{j=1}^{|I_2|} \psi^{(2)}_j(x) \left(\psi^{(2)}_j(y)\right)^{*}.
\end{align}
Since the states $\{\bm{\psi}^{(1)}_i\}_{i=1, \dots, |I_1|}$ are obtained by taking a linear combination of $\{\bm{\mu}_z \}_{z \in I_1}$, $\psi^{(1)}_i(x)$ is zero when $x \in \Lambda_2$, and similarly $\psi^{(2)}_j(x) = 0$ when $x \in \Lambda_1$.
Therefore, we have $\left(\mathsf{P}_0\right)_{x, y} = 0$ when $x \in \Lambda_1$ and $y \in \Lambda_2$. This, together with the Hermiticity of $\mathsf{P}_0$, implies that the matrix $\left(\left(\mathsf{P}_0\right)_{x, y}\right)_{x, y \in \Lambda_0}$ is reducible. 
\par
We then show that reducibility of $\left(\left(\mathsf{P}_0\right)_{x, y}\right)_{x, y \in \Lambda_0}$ implies that $\{\bm{\mu}_{z}\}_{z \in I}$ is not connected.
Assume that $\Lambda_0$ can be decomposed as $\Lambda_0 = \Lambda_1 \cup \Lambda_2$ with $\Lambda_1 \cap \Lambda_2 = \emptyset$, $\Lambda_1 \neq \emptyset$ and $\Lambda_2 \neq \emptyset$, in such a way that $\left(\mathsf{P}_0\right)_{x,y} = 0$ for any $x \in \Lambda_1$ and $y \in \Lambda_2$.
Let $\mathsf{P}_1$ and $\mathsf{P}_2$ be $|\Lambda| \times |\Lambda|$ matrices defined as 
\begin{align}
    \left(\mathsf{P}_j\right) =
    \begin{cases}
        \left(\mathsf{P}_0\right)_{x, y} \ \ \text{if} \ \ x, y \in \Lambda_j, \\ 
        0 \ \ \text{otherwise},
    \end{cases}
\end{align}
for $j = 1, 2$.
Both $\mathsf{P}_1$ and $\mathsf{P}_2$ are also projection matrices, and we see that $\mathsf{P}_0 = \mathsf{P}_1 + \mathsf{P}_2$ and $\mathsf{P}_1 \mathsf{P}_2 = \mathsf{P}_2 \mathsf{P}_1 = 0$.
Let $\mathsf{I}$ be the $|\Lambda| \times |\Lambda|$ identity matrix.
Because $\mathsf{I} - \mathsf{P}_j \geq 0$, we can repeat the same argument in Lemma~\ref{lemma:z basis} for $\mathsf{I} - \mathsf{P}_j$ for $j = 1, 2$.
Then, for $j = 1$ and $2$, we can take a subset $I_j \subset \Lambda$ with $|I_j| = \dim \ker (\mathsf{I} - \mathsf{P}_j)$ and a basis $\{ \bm{\mu}^{(j)}_z\}_{z \in I_j}$ of $\ker(\mathsf{I} - \mathsf{P}_j)$ in such a way that $\mu^{(j)}_{z}(z) \neq 0$ and $\mu^{(j)}_z(z') = 0$ for any $z' \in I_j \backslash \{z\}$.
Since $\bm{\mu}^{(j)}_z$ is an element of $\ker(\mathsf{I} - \mathsf{P}_j)$, it holds that 
\begin{align}
    \mathsf{P}_j \bm{\mu}^{(j)}_z = \bm{\mu}^{(j)}_z. \label{eq:projection1 and 2}
\end{align}
We can also see that $\langle \bm{\mu}^{(1)}_z, \bm{\mu}^{(2)}_{z'} \rangle = 0$ because 
\begin{align}
    \langle \bm{\mu}^{(1)}_z, \bm{\mu}^{(2)}_{z'} \rangle 
    &= \langle \mathsf{P}_1 \bm{\mu}^{(1)}_z, \mathsf{P}_2 \bm{\mu}^{(2)}_{z'} \rangle \nonumber \\ 
    & = \langle \bm{\mu}^{(1)}_z, \mathsf{P}_1 \mathsf{P}_2 \bm{\mu}^{(2)}_{z'} \rangle \nonumber \\
    & = 0,
\end{align}
where we have used Eq.~\eqref{eq:projection1 and 2}, $\mathsf{P}_1^\dag = \mathsf{P}_1$ and $\mathsf{P}_1 \mathsf{P}_2 = 0$.
From the relation $\mathsf{P}_1 \mathsf{P}_2 = 0$ and Eq. (\ref{eq:projection1 and 2}), we find 
\begin{align}
    \mathsf{P}_0 \bm{\mu}^{(j)}_z = \mathsf{P}_0 \mathsf{P}_j \bm{\mu}^{(j)}_z = \mathsf{P}_j^2 \bm{\mu}^{(j)}_z = \bm{\mu}^{(j)}_z,
\end{align}
which means that $\bm{\mu}^{(j)}_z \in \mathfrak{h}_0$. 
We note that $\Tr \mathsf{P}_0 = \dim{\mathfrak{h}_0} = D_0$, and similarly, $\Tr\mathsf{P}_1 = |I_1|$, and $\Tr\mathsf{P}_2 = |I_2|$.
Since $\Tr \mathsf{P}_0 = \Tr \mathsf{P}_1 + \Tr \mathsf{P}_2$, we have 
\begin{align}
    D_0 = |I_1| + |I_2|.
\end{align}
Therefore, the set of vectors $\{\bm{\mu}^{(1)}_z\}_{z \in I_1} \cup \{\bm{\mu}^{(2)}_{z'}\}_{z' \in I_2}$ is a basis of $\mathfrak{h}_0$ characterized by Lemma~\ref{lemma:z basis}.
From Eq.~\eqref{eq:projection1 and 2}, we see that $\mu_z^{(j)}(x)$ can be nonzero only if $x \in \Lambda_j$ for $j = 1$ and $2$, which implies that 
\begin{align}
    \bm{\mu}^{(1)}_z \nsim \bm{\mu}^{(2)}_{z'} \ \text{for any} \ z \in I_1 \ \text{and} \ z' \in I_2.
\end{align}
Hence, the basis $\{\bm{\mu}^{(1)}\}_{z \in I_1} \cup \{\bm{\mu}^{(2)}_{z'}\}_{z' \in I_2}$ is not connected.
\end{proof}
By combining Lemmas~\ref{lemma:SUn ferromagnetism and connectivity} and~\ref{lemma:connectivity and irreducibility}, we obtain the desired result, Theorem~\ref{thm:SUn flat band ferro}.
\par 
\section{Rigorous results for the ferromagnetic SU(\texorpdfstring{$N$}{N}) Kondo lattice model  \label{sec:Kondo}}
As an application of Theorem~\ref{thm:SUn flat band ferro}, we shall establish a theorem about SU($N$) ferromagnetism of the SU($N$) Kondo lattice model with a flat band.
\subsection{The SU(\texorpdfstring{$N$}{N}) spin operators \label{subsec:SU(N) spin in the schwinger rep}}
We first introduce the standard Schwinger fermion representation of the SU($N$) group.
Let the $N \times N$ matrices $G^{A}$ ($A = 1, \dots, N^2 - 1$) be the generators of the SU($N$) group such that 
\begin{align}
    &\Tr\left(G^A G^B \right) = \delta_{A, B}, \label{eq:normalization of generator}\\
    &[G^A, G^B] = \sum_{C=1}^{N^2 - 1} i f^{A, B, C} G^{C}, \label{eq:structure constant}\\
    &\sum_{A=1}^{N^2 - 1} \left(G^{A}\right)_{\alpha, \beta} \left(G^{A}\right)_{\mu, \nu} = \delta_{\alpha, \nu} \delta_{\beta, \mu} - \frac{1}{N} \delta_{\alpha, \beta} \delta_{\mu, \nu}, \label{eq:identity of matrix elements of gemerator}
\end{align}
where $f^{A, B, C}$ are structure constants of the $\mathfrak{su}(N)$ Lie algebra.
Let $\hat{c}_{\alpha}^\dag$ and $\hat{c}_{\alpha}$ be fermion creation and annihilation operators, where $\alpha = 1, \dots, N$.
They satisfy $\{\hat{c}_{\alpha}, \hat{c}_{\beta}\} = \{\hat{c}_{\alpha}^\dag, \hat{c}_{\beta}^\dag\} = 0$ and $\{\hat{c}_{\alpha}, \hat{c}_{\beta}^\dag\} = \delta_{\alpha, \beta}$.
With these operators, one can obtain the standard Schwinger fermion representation as 
\begin{align}
    \hat{S}^A = \sum_{\alpha, \beta = 1}^{N} \hat{c}_{\alpha}^\dag \left(G^{A}\right)_{\alpha, \beta} \hat{c}_{\beta}. \label{eq:std schwinger fermionic rep}
\end{align}
We see that 
\begin{align}
    [\hat{S}^A, \hat{S}^B] = \sum_{C=1}^{N^2 - 1} i f^{A, B, C} \hat{S}^C, 
\end{align}
which means that the operators $\hat{S}^A$ form a representation of the $\mathfrak{su}(N)$ Lie algebra.
In the following, we call the operators of the form Eq.~\eqref{eq:std schwinger fermionic rep} the SU($N$) spin operators.
\subsection{The Hamiltonian of the SU(\texorpdfstring{$N$}{N}) Kondo lattice model \label{subsec:hamiltonian of SUN klm}}
Here, we define the Hamiltonian of the SU($N$) Kondo lattice model.
Let $\Lambda$ be a finite lattice.
The operators $\hat{c}_{x, \alpha}^\dag$ and $\hat{c}_{x, \alpha}$ are creation and annihilation operators for itinerant fermions, which were introduced in Section~\ref{sec:SU(N) hubbard model and main result}.
The operators $\hat{f}_{x, \alpha}^\dag$ and $\hat{f}_{x, \alpha}$ are creation and annihilation operators for localized fermions at site $x \in \Lambda$ with color $\alpha$.
They satisfy 
\begin{align}
    \{\hat{f}_{x, \alpha}, \hat{f}_{y, \beta}\} &= \{\hat{f}_{x, \alpha}^\dag, \hat{f}_{y, \beta}^\dag\} = 0, \label{eq:ac f and f}\\
    \{\hat{f}_{x, \alpha}, \hat{f}_{y, \beta}^\dag\} &= \delta_{\alpha, \beta} \delta_{x, y}, \label{eq:ac f and fdag}\\
    \{\hat{c}_{x, \alpha}, \hat{f}_{y, \beta}\} &= \{\hat{c}_{x, \alpha}^\dag, \hat{f}_{y, \beta}^\dag\} = \{\hat{c}_{x, \alpha}, \hat{f}_{y, \beta}^\dag\} = \{\hat{c}_{x, \alpha}^\dag, \hat{f}_{y, \beta}\} = 0. \label{eq:ac c and f}
\end{align}
The number operator of localized fermion at site $x$ with color $\alpha$ is denoted by $\hat{n}^{(f)}_{x, \alpha} = \hat{f}_{x, \alpha}^\dag \hat{f}_{x, \alpha}$.
In the following, the number operator of the itinerant fermions $\hat{n}_{x, \alpha}$ is denoted by $\hat{n}^{(c)}_{x, \alpha} = \hat{c}_{x, \alpha}^\dag \hat{c}_{x, \alpha}$.
We also write $\hat{n}^{(c)}_x = \sum_{\alpha = 1}^N \hat{n}^{(c)}_{x, \alpha}$.
We denote the Fock space of the operators $\hat{c}_{x, \alpha}$ and $\hat{f}_{x, \alpha}$ by $\mathcal{H}'(\Lambda)$.
At each $x \in \Lambda$, we define the SU($N$) spin operators for itinerant and localized fermions by using the representation of Eq.~\eqref{eq:std schwinger fermionic rep} as
\begin{align}
    \hat{s}_x^A &= \sum_{\alpha, \beta =1}^N \hat{c}_{x, \alpha}^\dag \left(G^A\right)_{\alpha, \beta} \hat{c}_{x, \beta}, \\
    \hat{S}_x^A &= \sum_{\alpha, \beta =1}^N \hat{f}_{x, \alpha}^\dag \left(G^A\right)_{\alpha, \beta} \hat{f}_{x, \beta},
\end{align}
respectively.
The total SU($N$) spin operator is given by $\hat{S}_{\mathrm{tot}}^A = \sum_{x \in \Lambda} \hat{S}_{\mathrm{tot}, x}^A$, where $\hat{S}_{\mathrm{tot}, x}^A = \hat{s}_x^A + \hat{S}_x^A$.
The Hamiltonian of the SU($N$) Kondo lattice model is defined by
\begin{align}
    \hat{H}_{\mathrm{KLM}} = \sum_{\alpha=1}^N \sum_{x, y \in \Lambda} t_{x, y} \hat{c}_{x, \alpha}^\dag \hat{c}_{y, \alpha} + J_{\mathrm{K}} \sum_{x \in \Lambda} \sum_{A=1}^{N^2 - 1} \hat{s}_x^{A} \hat{S}_x^{A} \label{eq:klm hamiltonian}.
\end{align}
In the following, we study the case where the SU($N$) spins of itinerant and localized fermions are ferromagnetically coupled.
That is, we assume that $J_{\mathrm{K}} < 0$.
We also assume that there is exactly one fermion described by $\hat{f}_{x, \alpha}^\dag$ ($\alpha = 1, \dots, N$) localized at each site $x \in \Lambda$, i.e., at each site, there is the SU($N$) spin in the fundamental representation corresponding to a single box $\mysinglebox$.
In what follows, we only consider the subspace $\mathcal{W}(\Lambda) = \{\ket{\Psi} \in \mathcal{H}'(\Lambda)| \hat{n}_x^{(f)} \ket{\Psi} = \ket{\Psi} \ \text{for all} \ x \in \Lambda\}$, where $\hat{n}^{(f)}_x = \sum_{\alpha = 1}^N \hat{n}^{(f)}_{x, \alpha}$.
\par
The SU($N$) Kondo lattice model has the SU($N$) symmetry, which can be seen from the following relations
\begin{align}
    [\hat{H}_{\mathrm{KLM}}, \hat{S}_{\mathrm{tot}}^A] = 0 \ \ \text{for all} \ \ A=1, \dots, N^2 - 1.
\end{align} 
The SU($N$) symmetry can also be described by using the color raising and lowering operators.
The color raising and lowering operators are given as 
\begin{align}
    \hat{F}_{\mathrm{tot}}^{\alpha, \beta} = \sum_{x \in \Lambda} \left(\hat{c}_{x, \alpha}^\dag \hat{c}_{x, \beta} + \hat{f}_{x, \alpha}^\dag \hat{f}_{x, \beta}\right)\ \ \text{for} \ \ \alpha \neq \beta,
\end{align}
and the total number of itinerant and localized fermions with color $\alpha$ is given as 
\begin{align}
    \hat{F}_{\mathrm{tot}}^{\alpha, \alpha} = \sum_{x \in \Lambda} \left(\hat{n}_{x, \alpha}^{(c)} + \hat{n}^{(f)}_{x, \alpha}\right).
\end{align}
The SU($N$) symmetry of the model implies that 
\begin{align}
    [\hat{H}_{\mathrm{KLM}}, \hat{F}_{\mathrm{tot}}^{\alpha, \beta}] = 0.
\end{align}
The total number of itinerant and localized fermions is given by $\hat{N}_{\mathrm{tot}} = \sum_{\alpha = 1}^N \hat{F}^{\alpha, \alpha}$, and the Hamiltonian~\eqref{eq:klm hamiltonian} also commutes with $\hat{N}_{\mathrm{tot}}$.
We note that the Hamiltonian~\eqref{eq:klm hamiltonian} also satisfies 
\begin{align}
    [\hat{H}_{\mathrm{KLM}}, \hat{N}_c] = 0,
\end{align}
which means the model conserves the number of itinerant fermions described by $\hat{c}_{x, \alpha}^\dag$.
Therefore, the numbers of itinerant fermions and the number of localized fermions are independently conserved.
In the following, we denote the eigenvalue of $\hat{N}_{\mathrm{tot}}$, $\hat{N}_c$, and $\hat{F}_{\mathrm{tot}}^{\alpha, \alpha}$ by $N_{\mathrm{tot}}, N_c$, and $L_{\alpha}$, respectively.
Since it holds that $\hat{n}^{(f)}_x \ket{\Psi} = \ket{\Psi}$ in the subspace $\mathcal{W}(\Lambda)$, we have $\sum_{x \in \Lambda} \hat{n}^{(f)}_{x} \ket{\Psi} = |\Lambda| \ket{\Psi}$.
When the number of itinerant fermions is fixed to $N_c$, then $N_{\mathrm{tot}} = N_c + |\Lambda|$.
As in Section~\ref{sec:SU(N) hubbard model and main result}, we define the subspaces $\mathcal{H}'_{N_{\mathrm{tot}}}(\Lambda)$ and $\mathcal{H}'_{L_1, \dots, L_N}(\Lambda)$ by 
\begin{align}
    \mathcal{H}'_{N_{\mathrm{tot}}}(\Lambda) &= \{\ket{\Psi} \in \mathcal{H}'(\Lambda)~|~ \hat{N}_{\mathrm{tot}} \ket{\Psi} = N_{\mathrm{tot}} \ket{\Psi}\}, \\ 
    \mathcal{H}'_{L_1, \dots, L_N}(\Lambda) &= \{\ket{\Psi} \in \mathcal{H}'(\Lambda) ~|~ \hat{F}^{\alpha, \alpha}_{\mathrm{tot}} \ket{\Psi} = L_{\alpha} \ket{\Psi} \ \text{for all} \ \alpha = 1, \dots, N\}.
\end{align}
\par
To define the SU($N$) ferromagnetism, we again define the Casimir operator $\hat{C}_{ \mathrm{tot}, 2}$ of the SU($N$) group as 
\begin{align}
    \hat{C}_{\mathrm{tot}, 2} = \frac{1}{2} \left(\sum_{\alpha, \beta = 1}^N \hat{F}_{\mathrm{tot}}^{\alpha, \beta} \hat{F}_{\mathrm{tot}}^{\beta, \alpha} - \frac{\hat{N}_{\mathrm{tot}}^2}{N}\right).
\end{align}
Here we define the SU($N$) ferromagnetism of the SU($N$) Kondo lattice model as follows.
\begin{definition}
    Consider the Hamiltonian~\eqref{eq:klm hamiltonian} with a fixed $N_{\mathrm{tot}}$.
    We say that the model exhibits SU($N$) ferromagnetism if any ground state $\ket{\Psi_{\mathrm{GS}}}$ has the maximum eigenvalue of $\hat{C}_{\mathrm{tot}, 2}$ in $\mathcal{H}'_{N_{\mathrm{tot}}}(\Lambda)$ ,i.e.,
    \begin{align}
        \hat{C}_{\mathrm{tot}, 2} \ket{\Psi_{\mathrm{GS}}} = \frac{N_{\mathrm{tot}}(N-1)}{2} \left(\frac{N_{\mathrm{tot}}}{N} + 1\right) \ket{\Psi_{\mathrm{GS}}}.
    \end{align}
\end{definition}
\subsection{Theorem for the SU(\texorpdfstring{$N$}{N}) Kondo lattice model with a flat band \label{subsec: theorem for the SU(N) klm}}
In the following, we again assume that $\mathsf{T} \geq 0$ and $\mathfrak{h}_0 = \ker{\mathsf{T}}$ is not empty.
The dimension of $\mathfrak{h}_0$ is denoted by $D_0 > 0$.
With the basis introduced in Lemma~\ref{lemma:z basis}, we can define the set of operators $\hat{a}_{z, \alpha}^\dag$ as in Eq.~\eqref{eq:a operator}.
We then define the fully polarized states with the same color $\alpha$ for itinerant and localized fermions as 
\begin{align}
    \ket{\Psi_{\mathrm{all} \ \alpha}} = \left(\prod_{z \in I} \hat{a}_{z, \alpha}^\dag\right) \left(\prod_{x \in \Lambda} \hat{f}_{x, \alpha}^\dag\right) \ket{\Psi_{\mathrm{vac}}}, \label{eq:all aplha klm}
\end{align}
where $\ket{\Psi_{\mathrm{vac}}}$ is the normalized vacuum state for $\hat{c}_{x, \alpha}$ and $\hat{f}_{x, \alpha}$ operators.
We also define the state of the form,
\begin{align}
    \ket{\Psi_{L_1, \dots, L_N}} = \left(\hat{F}_{\mathrm{tot}}^{N, 1}\right)^{L_N} \cdots \left(\hat{F}_{\mathrm{tot}}^{2, 1}\right)^{L_2} \ket{\Psi_{\mathrm{all} \ 1}}, \label{eq:fully polarized state klm}
\end{align}
where $L_1 = D_0 +|\Lambda| - \sum_{\alpha=2}^{N} L_{\alpha}$.
We also call the states of the form~\eqref{eq:fully polarized state klm} fully polarized states.
In the same manner as Eq.~\eqref{eq:phi all alpha has maximum C2}, it can be checked that 
\begin{align}
    \hat{C}_{\mathrm{tot}, 2} \ket{\Psi_{\mathrm{all} \ \alpha}} = \frac{N_{\mathrm{tot}}(N-1)}{2} \left(\frac{N_{\mathrm{tot}}}{N} + 1\right) \ket{\Psi_{\mathrm{all} \ \alpha}},
\end{align}
where $N_{\mathrm{tot}} = D_0 + |\Lambda|$.
We can see that the fully polarized states of the form~\eqref{eq:fully polarized state klm} also have the same eigenvalue of $\hat{C}_{\mathrm{tot, 2}}$, and thus all the fully polarized states have the maximum eigenvalue of $\hat{C}_{\mathrm{tot}, 2}$.

Using Theorem~\ref{thm:SUn flat band ferro}, we can prove the following theorem:
\begin{theorem} \label{thm:SUn flat band ferro klm}
    Consider the Hamiltonian~\eqref{eq:klm hamiltonian} with $N_{\mathrm{tot}} = D_0 + |\Lambda|$ and $J_{\mathrm{K}} < 0$. 
    The ferromagnetic SU($N$) Kondo lattice model exhibits SU($N$) ferromagnetism when the subset $\Lambda_0$ introduced in Theorem~\ref{thm:SUn flat band ferro} satisfies $\Lambda_0 = \Lambda$, and the matrix $\left(\left(\mathsf{P}_0\right)_{x, y}\right)_{x, y \in \Lambda_0}$ is irreducible.
\end{theorem}
Before we proceed with the proof of Theorem~\ref{thm:SUn flat band ferro klm}, let us remark that the ferromagnetic SU($N$) Kondo lattice model on the delta chain shown in Fig.~\ref{fig:delta chain} exhibits SU($N$) ferromagnetism.
As mentioned earlier, in the delta chain, $\Lambda_0 = \Lambda$ holds true.
Since the irreducibility of $\mathsf{P}_0$ is also proven, according to Theorem~\ref{thm:SUn flat band ferro klm}, this is an example of the ground state being SU($N$) ferromagnetic.

\subsection{Proof of Theorem~\ref{thm:SUn flat band ferro klm} \label{subsec:proof of thm2}}
We decompose the Hamiltonian~\eqref{eq:klm hamiltonian} as
\begin{align}
    \hat{H}_ {\mathrm{KLM}} &= \hat{H}_{\mathrm{Hub}} + J_{\mathrm{K}} \left(1 - \frac{1}{N}\right) \hat{N}_c + \sum_{x \in \Lambda} \hat{V}_x, \label{eq:decomposition of klm}
\end{align}
where $\hat{H}_{\mathrm{Hub}}$ is the Hamiltonian of the SU($N$) Hubbard model, 
\begin{align}
    \hat{H}_{\mathrm{Hub}} &= \sum_{\alpha=1}^{N} \sum_{x, y \in \Lambda} t_{x, y} \hat{c}_{x, \alpha}^\dag \hat{c}_{y, \alpha} + U \sum_{\alpha < \beta} \sum_{x \in \Lambda} \hat{n}^{(c)}_{x, \alpha} \hat{n}^{(c)}_{x, \beta}, 
\end{align}
with $U > 0$.
The local interaction $\hat{V}_x$ is defined by
\begin{align}
    \hat{V}_x &= -U \sum_{\alpha < \beta} \hat{n}^{(c)}_{x, \alpha} \hat{n}^{(c)}_{x, \beta}
    - J_{\mathrm{K}} \left(1 - \frac{1}{N}\right) \hat{n}^{(c)}_x + J_{\mathrm{K}} \sum_{A=1}^{N^2 - 1} \hat{s}_x^{A} \hat{S}_x^{A}.
\end{align}
In the proof, we use the following lemma.
\begin{lemma} \label{lemma: positive demidefinitenss of Vx}
    The local interaction $\hat{V}_x$ is positive semidefinite when $|J_{\mathrm{K}}|/U > N/2$.
\end{lemma}
\begin{proof}
To prove this, we study the eigenvalues of $\hat{V}_x$, which is denoted by $V_x$ in the following.
We can express $\hat{V}_x$ as
\begin{align}
    \hat{V}_x &= 
    -\frac{U}{2} \hat{n}^{(c)}_x (\hat{n}^{(c)}_x - 1) 
    - J_{\mathrm{K}} \left(1 -\frac{1}{N}\right) \hat{n}^{(c)}_x \nonumber \\ 
    &+ J_{\mathrm{K}}\left(\frac{1}{2} \sum_{A = 1}^{N^2 - 1} \hat{S}_{\mathrm{tot}, x}^A \hat{S}_{\mathrm{tot}, x}^A - \frac{1}{2} \sum_{A=1}^{N^2 - 1} \hat{s}_x^A \hat{s}_x^A - \frac{1}{2} \sum_{A=1}^{N^2 - 1} \hat{S}_x^A \hat{S}_x^A 
    \right).
\end{align}
Since we consider the subspace $\mathcal{W}(\Lambda)$, the operator $\frac{1}{2} \sum_{A=1}^{N^2 - 1} \hat{S}_x^A \hat{S}_x^A $ is the quadratic Casimir operator for the fundamental representation, and its eigenvalue is 
\begin{align}
    C_2\left( \!
    \mysinglebox
    \!\right)
    = \frac{1}{2N}(N^2 - 1) \label{eq:casimir one box}.
\end{align}
Because $\hat{n}^{(c)}_x$ commutes with $\hat{V}_x$, we can fix the number of itinerant fermions, which is denoted by $n^{(c)}$ ($n^{(c)} = 0, \dots, N$).
In this sector, the operator $\frac{1}{2} \sum_{A=1}^{N^2 - 1} \hat{s}_x^A \hat{s}_x^A$ is the quadratic Casimir operator for the representation 
\begin{align}
    \myfullyantisymmetrizedbox{$n^{(c)}$},
\end{align}
and its eigenvalue is given by 
\begin{align}
    C_2\left(\myfullyantisymmetrizedbox{$n^{(c)}$}\right) = \frac{N+1}{2N} n^{(c)} (N - n^{(c)}).
\end{align}
Therefore, when $n^{(c)}$ is fixed, the local interaction $\hat{V}_x$ acts as 
\begin{align}
    \hat{V}_x &= -\frac{U}{2} n^{(c)} (n^{(c)} - 1) 
    - J_{\mathrm{K}} \left(1 - \frac{1}{N}\right) n^{(c)} \nonumber \\ 
    &+ J_{\mathrm{K}} \left(
    \frac{1}{2} \sum_{A = 1}^{N^2 - 1} \hat{S}_{\mathrm{tot}, x}^A \hat{S}_{\mathrm{tot}, x}^A 
    - \frac{N+1}{2N} n^{(c)} (N - n^{(c)}) 
    - \frac{1}{2N}(N^2 - 1) 
    \right).
\end{align}
When the number of itinerant fermions is $n^{(c)}$, there are two possible total SU($N$) spins in the representations 
\begin{align}
    \myfullyantisymmetrizedbox{$n^{(c)} + 1$}, 
    \myantisymmetrizedandsimmetrizedbox{$n^{(c)}$},
\end{align}
and the eigenvalues of the quadratic Casimir operator for these representations are 
\begin{align}
    C_2\left(\!
    \myfullyantisymmetrizedbox{$n^{(c)} + 1$}
    \!\right) = \frac{N+1}{2N} (n^{(c)} + 1) (N - n^{(c)} - 1),
\end{align}
and
\begin{align}
    C_2\left(\!
    \myantisymmetrizedandsimmetrizedbox{$n^{(c)}$}
    \!\right) = n^{(c)} + 1 + \frac{N+1}{2N} (n^{(c)} + 1) (N - n^{(c)} - 1).
\end{align}
When $n^{(c)} = 0$, there is only one possible total SU($N$) spin described by $\mysinglebox$.
When $n^{(c)} = N$, the itinerant fermions form the SU($N$) singlet, in which the SU($N$) spin transforms in the trivial representation corresponding to the Young diagram with no box.
Thus, in this case, there is also one possible total SU($N$) spin in the representation $\mysinglebox$.
The eigenvalue of the quadratic Casimir operator for the representation $\mysinglebox$ is given by Eq.~\eqref{eq:casimir one box}.
In this way, we find 
\begin{align}
    V_x &= 0 \ \ \text{for} \ n^{(c)} = 0, \label{eq:vx, nc=0} \\
    V_x &= - \frac{U}{2} n^{(c)}(n^{(c)} - 1) - 2 J_{\mathrm{K}} n^{(c)}  \ \ \text{for} \ \myfullyantisymmetrizedbox{$n^{(c)} + 1$} \ (1 \leq n^{(c)} \leq N-1), \label{eq:vx, nc, antiferro}\\ 
    V_x &= - \frac{U}{2} n^{(c)}(n^{(c)} - 1) - J_{\mathrm{K}} (n^{(c)} - 1) \ \ \text{for} \ \myantisymmetrizedandsimmetrizedbox{$n^{(c)}$} \ \ (1 \leq n_c \leq N-1), \label{eq:vx, nc, ferro}\\ 
    V_x &= -\frac{U}{2}N(N-1) - J_{\mathrm{K}} (N-1)\ \ \text{for} \ n^{(c)} = N. \label{eq:vx, nc=N}
\end{align}
When $n^{(c)} = 0$, we see that $V_x = 0$.
For $n^{(c)} = 1, \dots, N - 1$, when $J_{\mathrm{K}} < 0$,  the eigenvalue of Eq.~\eqref{eq:vx, nc, antiferro} is strictly larger than that of Eq.~\eqref{eq:vx, nc, ferro}.
One can easily see the eigenvalues~\eqref{eq:vx, nc, ferro} cannot be negative when $|J_{\mathrm{K}}|/U > N/2$.
For $n^{(c)} = N$, the eigenvalue~\eqref{eq:vx, nc=N} is positive when $|J_{\mathrm{K}}|/U > N/2$.
Thus, all the eigenvalues of $\hat{V}_x$ are greater than or equal to zero, which implies that $\hat{V}_x \geq 0$.
\end{proof}
\par
We note that in the decomposition~Eq.~\eqref{eq:decomposition of klm}, one can take the parameter $U > 0$ arbitrarily.
For a given $J_{\mathrm{K}}$, there always exists $U$ such that $|J_{\mathrm{K}}|/U > N/2$.
In the following proof, $U$ is assumed to be such that $|J_{\mathrm{K}}|/U > N/2$, and hence $\hat{V}_x \geq 0$.
\par
Here we show that the fully polarized states~\eqref{eq:all aplha klm} and~\eqref{eq:fully polarized state klm} are ground states.
Since $\hat{H}_{\mathrm{Hub}} \geq 0$ and $\hat{V}_x \geq 0$, we get a lower bound
\begin{align}
    \hat{H}_{\mathrm{KLM}} \geq - D_0 |J_{\mathrm{K}}| \left(1 - \frac{1}{N}\right) , \label{eq:inequality of klm}
\end{align}
where we have used $N_c = D_0$ because we consider the Hilbert space $\mathcal{H}'_{D_0 + |\Lambda|}(\Lambda)$.
We can easily check that the fully polarized state $\ket{\Psi_{\mathrm{all}\ 1}}$ satisfies 
\begin{align}
    \hat{H}_{\mathrm{KLM}}\ket{\Psi_{\mathrm{all}\ 1}} = -D_0|J_{\mathrm{K}}| \left(1 - \frac{1}{N}\right) \ket{\Psi_{\mathrm{all}\ 1}},
\end{align}
and hence the state $\ket{\Psi_{\mathrm{all}\ 1}}$ is a ground state of $\hat{H}_{\mathrm{KLM}}$.
Due to the SU($N$) symmetry, all the fully polarized states are ground states of $\hat{H}_{\mathrm{KLM}}$.
\par 
In the rest of the proof, we prove that there are no other ground states.
Let $\ket{\Psi_{\mathrm{GS}}}$ be an arbitrary ground state, which satisfies 
\begin{align}
    \hat{H}_{\mathrm{KLM}} \ket{\Psi_{\mathrm{GS}}} = -D_0|J_{\mathrm{K}}| \left(1 - \frac{1}{N}\right) \ket{\Psi_{\mathrm{GS}}}.
\end{align}
With the decomposition~\eqref{eq:decomposition of klm}, noting that $\hat{H}_{\mathrm{Hub}} \geq 0$ and $\hat{V}_x \geq 0$, we find 
\begin{align}
    \hat{H}_{\mathrm{Hub}} \ket{\Psi_{\mathrm{GS}}} = 0, \label{eq:klm hubbard condition}
\end{align}
and 
\begin{align}
    \hat{V}_x \ket{\Psi_{\mathrm{GS}}} = 0 \ \ \text{for all} \ \ x \in \Lambda. \label{eq:klm vx condition}
\end{align}
We first consider the condition~\eqref{eq:klm hubbard condition}.
According to Theorem~\ref{thm:SUn flat band ferro}, the ground state $\ket{\Psi_{\mathrm{GS}}}$ can be written as 
\begin{align}
    \ket{\Psi_{\mathrm{GS}}} = \sum_{\bm{\alpha}} \sum_{\bm{\beta}} C(\bm{\alpha}, \bm{\beta}) \left(\prod_{z \in I} \hat{a}_{z, \alpha_z}^\dag\right) \left(\prod_{x \in \Lambda} \hat{f}_{x, \beta_x}^\dag\right) \ket{\Psi_{\mathrm{vac}}},
\end{align}
where $\bm{\alpha} = \left(\alpha_z\right)_{z \in I}$ is a color configuration of itinerant fermions over the subset $I$, and $\bm{\beta} = \left(\beta_{x}\right)_{x \in \Lambda}$ is a color configuration of localized fermions over the lattice $\Lambda$.
The coefficients $C(\bm{\alpha}, \bm{\beta})$ must be symmetric under permutations of $\bm{\alpha}$.
\par
Then we consider the condition~\eqref{eq:klm vx condition}.
Using the local constraint $\hat{n}^{(f)}_x \ket{\Psi_{\mathrm{GS}}} = \ket{\Psi_{\mathrm{GS}}}$ and Eq.~\eqref{eq:identity of matrix elements of gemerator}, we have
\begin{align}
    \hat{V}_x \ket{\Psi_{\mathrm{GS}}} 
    = \left(-\frac{U}{2} \hat{n}_x^{(c)}(\hat{n}_x^{(c)} - 1) + 
    |J_{\mathrm{K}}| \sum_{\alpha < \beta} \left(\hat{f}_{x, \beta}^\dag \hat{c}_{x, \alpha}^\dag - \hat{f}_{x, \alpha}^\dag \hat{c}_{x, \beta}^\dag\right)
    \left(\hat{c}_{x, \alpha} \hat{f}_{x, \beta} - \hat{c}_{x, \beta} \hat{f}_{x, \alpha}\right)\right) \ket{\Psi_{\mathrm{GS}}}.
\end{align}
Since $\hat{n}_x^{(c)}(\hat{n}_x^{(c)} - 1) \ket{\Psi_{\mathrm{GS}}} = 0$, the ground state satisfies
\begin{align}
    |J_{\mathrm{K}}| \sum_{\alpha < \beta} \left(\hat{f}_{x, \beta}^\dag \hat{c}_{x, \alpha}^\dag - \hat{f}_{x, \alpha}^\dag \hat{c}_{x, \beta}^\dag\right)
    \left(\hat{c}_{x, \alpha} \hat{f}_{x, \beta} - \hat{c}_{x, \beta} \hat{f}_{x, \alpha}\right) \ket{\Psi_{\mathrm{GS}}} = 0,
\end{align}
which leads to 
\begin{align}
    \left(\hat{c}_{x, \alpha} \hat{f}_{x, \beta} - \hat{c}_{x, \beta} \hat{f}_{x, \alpha}\right) \ket{\Psi_{\mathrm{GS}}} = 0, \label{eq:cf - cf=0}
\end{align}
for any $x \in \Lambda$ and $\alpha \neq \beta$.
Using the anticommutation relations~\eqref{eq:anticommutation rel c and a} and~\eqref{eq:ac c and f}, from Eq.~\eqref{eq:cf - cf=0}, we obtain
\begin{align}
    &(-1)^{D_0} \sum_{z \in I} 
    \sum_{\substack{\bm{\alpha} \\ \alpha_z = \alpha}}
    \sum_{\substack{\bm{\beta} \\ \beta_x = \beta}}
    \mathrm{sgn}(z; I) \mu_z(x) 
    \left(C(\bm{\alpha}, \bm{\beta}) - C(\bm{\alpha}|_{\alpha_z = \beta}, \bm{\beta}|_{\beta_x = \alpha})\right) \nonumber \\
    &\times \left(\prod_{z' \in I \backslash\{z\}} \hat{a}_{z', \alpha_{z'}}^\dag\right) 
    \left(\prod_{x' \in \Lambda \backslash\{x\}} \hat{f}_{x', \beta_{x'}}^\dag\right) \ket{\Psi_{\mathrm{vac}}} = 0,
\end{align}
where $\bm{\alpha}|_{\alpha_z = \beta}$ is the color configuration obtained from $\bm{\alpha}$ by replacing $\alpha_z$ with $\beta$, and similarly, $\bm{\beta}|_{\beta_x = \alpha}$ is the color configuration obtained from $\bm{\beta}$ by replacing $\beta_x$ with $\alpha$.
The function $\mathrm{sgn}(z; I)$ is a sign factor arising from exchanges of the fermion operators.

Since all the states in the sum are linearly independent and $\alpha$ and $\beta$ are arbitrary, we get
\begin{align}
    \mu_z(x) \left(C(\bm{\alpha}, \bm{\beta}) - C(\bm{\alpha}|_{\alpha_z = \beta_x}, \bm{\beta}|_{\beta_x = \alpha_z}) \right)= 0,
\end{align}
for any $x \in \Lambda$, $z \in I$, $\bm{\alpha}$, and $\bm{\beta}$.
By assumption, it holds that $\Lambda_0 = \Lambda$, and hence, for any $x \in \Lambda$, there exists $z_0 \in I$ such that $\mu_{z_0}(x) \neq 0$.
For such $z_0 \in I$, we obtain
\begin{align}
    C(\bm{\alpha}, \bm{\beta}) = C(\bm{\alpha}|_{\alpha_{z_0} = \beta_x}, \bm{\beta}|_{\beta_x = \alpha_{z_0}})  \ \ \text{for any} \ \bm{\alpha} \ \text{and} \ \bm{\beta}.
\end{align}
Noting that $C(\bm{\alpha}, \bm{\beta})$ is symmetric under the permutations of $\bm{\alpha}$ and $x \in \Lambda$ is arbitrary, we see that 
\begin{align}
    C(\bm{\alpha}, \bm{\beta}) = C(\bm{\alpha}|_{\alpha_z = \beta_x}, \bm{\beta}|_{\beta_x = \alpha_z}), \label{eq:klm swap alpha and beta}
\end{align}
for all $x \in \Lambda$ and $z \in I$.
Using Eq.~\eqref{eq:klm swap alpha and beta}, we also find that $C(\bm{\alpha}, \bm{\beta}) = C(\bm{\alpha}, \bm{\beta}_{x \leftrightarrow y})$, where $\bm{\beta}_{x \leftrightarrow y}$ is obtained from $\bm{\beta}$ by swapping $\beta_x$ and $\beta_y$. 
Since any permutation of can be obtained by repeatedly swapping two colors, we have
\begin{align}
    C(\bm{\alpha}, \bm{\beta}) = C(\bm{\alpha}', \bm{\beta}'), \label{eq:Calphabeta=Calpha'beta', klm}
\end{align}
where $(\bm{\alpha}', \bm{\beta}')$ is a permutation of $(\bm{\alpha}, \bm{\beta})$ with $(\bm{\alpha}, \bm{\beta})$ being a color configuration 
over $I$ and $\Lambda$. 
Consequently, when we fix $L_{\alpha}$ for all $\alpha=1, \dots, N$, the coefficient $C(\bm{\alpha}, \bm{\beta})$ is a constant, and hence the ground state is unique in $\mathcal{H}'_{L_1, \dots, L_N}(\Lambda)$.
\par 
Finally, we can also check that the unique ground state in $\mathcal{H}'_{L_1, \dots, L_N}(\Lambda)$ is the fully polarized state $\ket{\Psi_{L_1, \dots, L_N}}$ in a similar way as in Section~\ref{sec:SU(N) hubbard model and main result}.
Here we introduce a word $w' = (w'_1, \dots, w'_{D_0}, w'_{D_0 + 1}, \dots, w'_{D_0 + |\Lambda|})$ whose length is $D_0 + |\Lambda|$ and denote the set of words for which $|w'_{\alpha}| = L_{\alpha}$ by $W'(L_1, \dots, L_N) = \{w' \, | \,  |w'_{\alpha}| = L_{\alpha} \ \text{for all} \ \alpha\}$.
The ground state in $\mathcal{H}'_{L_1, \dots, L_N}(\Lambda)$ satisfying Eq.~\eqref{eq:Calphabeta=Calpha'beta', klm} is written as 
\begin{align}
    \ket{\tilde{\Psi}_{L_1, \dots, L_N}} = \sum_{w' \in W'(L_1, \dots, L_N)} \hat{a}_{z_1, w'_1}^\dag  \cdots \hat{a}_{z_{D_0}, w'_{D_0}}^\dag \hat{f}_{1, w'_{D_0 + 1}}^\dag \cdots \hat{f}_{|\Lambda|, w'_{D_0 + |\Lambda|}}^\dag \ket{\Psi_{\mathrm{vac}}},
\end{align}
where we have labeled each site $x \in \Lambda$ by integers as $x = 1, \dots, |\Lambda|$.
With the commutation relations $[\hat{F}_{\mathrm{tot}}^{\alpha, \beta}, \hat{a}_{z, \gamma}^\dag] = \delta_{\alpha, \gamma} \hat{a}_{z, \beta}^\dag$, and $[\hat{F}_{\mathrm{tot}}^{\alpha, \beta}, \hat{f}_{x, \gamma}^\dag] = \delta_{\alpha, \gamma} \hat{f}_{x, \beta}^\dag$, we find that 
\begin{align}
    \left(\hat{F}_{\mathrm{tot}}^{2, 1}\right)^{L_2} \ket{\Psi_{\mathrm{all} \ 1}} = L_2 ! \sum_{w' \in W'(D_0 + |\Lambda| - L_2, L_2, 0, \dots, 0)} \hat{a}_{z_1, w'_1}^\dag  \cdots \hat{a}_{z_{D_0}, w'_{D_0}}^\dag \hat{f}_{1, w'_{D_0 + 1}}^\dag \cdots \hat{f}_{|\Lambda|, w'_{D_0 + |\Lambda|}}^\dag \ket{\Psi_{\mathrm{vac}}}.
\end{align}
By repeating the same calculations, we see that the states $\ket{\tilde{\Psi}_{L_1, \dots, L_N}}$ and $\ket{\Psi_{L_1, \dots, L_N}}$ are the same up to a normalization.
Therefore, the unique ground state $\ket{\tilde{\Psi}_{L_1, \dots, L_N}}$ in $\mathcal{H}'_{D_0 + |\Lambda|}(\Lambda)$ is indeed the fully polarized state.

\subsection{Remark}
Here we would like to comment on the stability of the flat-band ferromagnetism for the SU($N$) Kondo lattice model.
In previous studies, rigorous results regarding the stability of the flat-band ferromagnetism in the SU(2) Hubbard model have been obtained~\cite{tasaki1994stability,tasaki1995ferromagnetism,tasaki1996stability,tanaka2003stability,tanaka2018ferromagnetism,tasaki2020physics}.
The extension of these results to the SU($N$) case has also been made for a particular class of systems~\cite{tamura2019ferromagnetism,tamura2021ferromagnetism}.
By combining these results with the technique for the proof of Theorem~\ref{thm:SUn flat band ferro klm}, we can also discuss the stability of flat-band ferromagnetism for the ferromagnetic SU($N$) Kondo lattice model in a mathematically rigorous way.
To illustrate this, let us consider a sufficiently large interaction strength $U$ such that the SU($N$) Hubbard model exhibits SU($N$) ferromagnetism and assume that the other parameters are also in the range where SU($N$) ferromagnetism occurs.
For the Kondo coupling $J_{\mathrm{K}}$ such that $|J_{\mathrm{K}}|/U > N/2$, we can repeat the same argument as in the proof of Theorem~\ref{thm:SUn flat band ferro klm}.
In this way, one can establish rigorous results on the stability of ferromagnetism in the ferromagnetic SU($N$) Kondo lattice model with a nearly flat band.

\section{Conclusion and remark \label{sec:conclusion}}
In this paper, we have established rigorous results on flat-band ferromagnetism for the SU($N$) Hubbard model and the ferromagnetic SU($N$) Kondo lattice model.
For the former, we found the necessary and sufficient condition for the ground state to exhibit SU($N$) ferromagnetism when the number of particles is equal to the degeneracy of the lowest-energy single-particle states.
The condition says that the irreducibility of the projection matrix onto the space of the lowest-energy single-particle states is equivalent to the presence of SU($N$) ferromagnetism.
We also showed that this general theory could be applied to the ferromagnetic SU($N$) Kondo lattice model with the hopping term of itinerant fermions that has a flat band at the bottom.
Specifically, we considered the case in which the number of itinerant fermions is the same as the degeneracy of the flat band, and each site is occupied by one localized fermion.
We then proved that the model with a nonzero ferromagnetic Kondo coupling exhibits SU($N$) ferromagnetism when certain conditions with respect to the hopping matrix are satisfied.
In our setup, we have considered situations where the number of fermions is fixed to fully occupy the lowest band, resulting in an insulating system.
However, it will be intriguing to investigate ferromagnetism with different fillings, at which the system is expected to be metallic.
Nevertheless, even in the conventional SU(2) case, this remains inherently difficult and still challenging.
Therefore, addressing metallic ferromagnetism may require new mathematical methods and physical perspectives.
\par
In addition, we discussed the ferromagnetic SU($N$) Kondo lattice model in Theorem~\ref{thm:SUn flat band ferro klm}.
While the presence of ferromagnetic interaction is physically natural in the context of ultracold atomic systems, considering antiferromagnetic interaction is also of interest.
However, the approach employed in our proof is not readily applicable to the antiferromagnetic case because the proof for the Lemma~\ref{lemma: positive demidefinitenss of Vx} does not work. 
Therefore, another method will be required to obtain rigorous results for the antiferromagnetic SU($N$) Kondo lattice model.
For example, for the antiferromagnetic SU(2) Kondo lattice model, a rigorous proof has been given that the model with one conduction electron exhibits an incomplete ferromagnetic order~\cite{sigrist1991rigorous}.
It might be possible to extend the result to the general SU(N) cases; however, this particular issue lies beyond the scope of our current study and is left for future investigation.
\par
Finally, we discuss a potential extension of Theorem~\ref{thm:SUn flat band ferro klm}.
As we discussed, the SU($N$) Kondo lattice model is expected to be experimentally realizable with ultracold atomic gases.
While the Kondo lattice model neglects the on-site interaction among itinerant fermions, in principle, such interaction can be present.
Thus it is worth investigating flat-band ferromagnetism for the models with both the on-site interaction $U > 0$ and the ferromagnetic Kondo coupling $J_{\mathrm{K}} < 0$.
Since it has been rigorously proved that the SU($N$) Hubbard 
and ferromagnetic 
Kondo lattice model with a flat band exhibit SU($N$) ferromagnetism when certain conditions are met, it would be possible to establish rigorous results on flat-band ferromagnetism in the presence of both interactions.


\section*{Acknowledgments}
K.T. was supported by JSPS KAKENHI Grant No. 21J11575. 
H.K. was supported in part by JSPS Grant-in-Aid for Scientific Research on Innovative Areas No. JP23H01086, JSPS KAKENHI Grant No. JP18K03445, Grant-in-Aid for Transformative Research Areas A “Extreme Universe” No. JP21H05191, and the Inamori Foundation.

\section*{References}
\bibliography{reference}
\bibliographystyle{iopart-num}
\end{document}